\documentclass[10pt,twocolumn]{IEEEtran}

\usepackage[pdftex]{graphicx}
\usepackage{amssymb}
\usepackage{amsthm}
\usepackage{amsmath}
\usepackage[caption=false,font=footnotesize]{subfig}
\usepackage{color}
\usepackage{cite}
\usepackage{tikz}
\usepackage{bm}
\usepackage{cite}

\newcommand{\mbf}{\bm}
\newcommand{\mbb}{\mathbb}
\newcommand{\bs}{\bm}

\newcommand{\bkt}[1]{\left( #1\right)}
\newcommand{\inner}[1]{\left\langle #1\right\rangle}
\newcommand{\brac}[1]{\left\{ #1\right\}}
\newcommand{\abs}[1]{\left| #1\right|}

\newcommand{\Tf}{\mathcal{T}(\hat{f}\,)}

\newcommand{\Tfn}{\mathcal{T}(\hat{f}_n\,)}
\newcommand{\Tg}{\mathcal{T}(\hat{g}\,)}
\newcommand{\Txf}{\mathcal{T}_x(\hat{f}\,)}
\newcommand{\Tyf}{\mathcal{T}_y(\hat{f}\,)}

\newcommand{\ER}{\bm E_{\text{row}}}
\newcommand{\EL}{\bm E_{\text{col}}}
\newcommand{\GR}{\bm G}

\newtheorem{thm}{Theorem}

\newtheorem{lem}[thm]{Lemma}
\newtheorem{prop}[thm]{Proposition}
\newtheorem{defn}[thm]{Definition}
\newtheorem*{thm*}{Theorem}
\newtheorem*{rem*}{Remark}

\usepackage{xcolor}
\usepackage{pgfplots} 
\pgfplotsset{compat=newest} 
\pgfplotsset{plot coordinates/math parser=false} 
\pgfplotsset{ 
  legend style =
  {font=\small\sffamily},
  label style = {font=\footnotesize\sffamily},
	tick label style = {font=\footnotesize}
}
\newlength\figureheight 
\newlength\figurewidth 
\title{Convex recovery of continuous domain piecewise constant images from non-uniform Fourier samples}
\author{Greg Ongie,~\IEEEmembership{Member,~IEEE}, Sampurna Biswas,~\IEEEmembership{Student Member,~IEEE}, Mathews~Jacob*,~\IEEEmembership{Senior Member,~IEEE}
\thanks{
G. Ongie is with the Department of EECS, University of Michigan, Ann Arbor, MI 48108 USA. S. Biswas and M. Jacob are with the Department of Electrical and Computer Engineering, University of Iowa, Iowa City, IA, 52245 USA (e-mail: gongie@umich.edu; sampurna-baswas@uiowa.edu; mjacob@uiowa.edu)}
\thanks{This work is supported by grants NIH 1R01EB019961-01A1 and ONR N00014-13-1-0202.}}

\begin{document}
%
\maketitle
%

\begin{abstract}
We consider the recovery of a continuous domain piecewise constant image from its non-uniform Fourier samples using a convex matrix completion algorithm. We assume the discontinuities/edges of the image are localized to the zero level-set of a bandlimited function. This assumption induces linear dependencies between the Fourier coefficients of the image, which results in a two-fold block Toeplitz matrix constructed from the Fourier coefficients being low-rank. The proposed algorithm reformulates the recovery of the unknown Fourier coefficients as a structured low-rank matrix completion problem, where the nuclear norm of the matrix is minimized subject to structure and data constraints. We show that exact recovery is possible with high probability when the edge set of the image satisfies an incoherency property. We also show that the incoherency property is dependent on the geometry of the edge set curve, implying higher sampling burden for smaller curves. This paper generalizes recent work on the super-resolution recovery of isolated Diracs or signals with finite rate of innovation to the recovery of piecewise constant images. 

\end{abstract}

\section{Introduction}
\label{sec:intro}
The direct recovery of continuous domain signals by convex optimization is emerging as a powerful alternative to traditional discrete domain compressed sensing \cite{bhaskar2011atomic,candes2014towards,chen2014robust}. The ability of these continuous domain ``off-the-grid'' schemes to minimize discretization errors makes them attractive in practical applications, where only the low-pass measurements of the signal are available. The history of such continuous domain signal recovery algorithms dates back to Prony \cite{stoica1997introduction}, where the recovery of a linear combination of exponentials from uniform samples is considered. Prony-like algorithms recover the signal by estimating an annihilating polynomial whose zeros correspond to the frequencies of the exponentials. Work by Liang et al.~\cite{haacke1989super,haacke1989constrained} and the finite rate of innovation (FRI) framework \cite{vetterli2002sampling}
 extended Prony-like methods to recover more general signals that reduce to a sparse linear combination of Dirac delta functions under an appropriate transformation (e.g., differential operators, convolution). Recently, several authors have further extended FRI methods to recover such signals from their non-uniform Fourier samples \cite{chen2014robust,ALOHAarxiv,ALOHAisbi,haldar2014low,sampta2015} by exploiting the low-rank structure of an enhanced matrix built from Fourier data (e.g., a Hankel matrix in 1-D). Recovery guarantees exists for certain classes of these signals when the singularities are isolated and well-separated \cite{candes2014towards,chen2014robust,ye2016compressive}. 

The signal models discussed above have limited flexibility in exploiting the extensive additional structure present in multidimensional imaging problems. In particular, the edges in multidimensional images are connected and can be modeled as smooth curves or surfaces. While discrete image representations to capture this structure have been the subject extensive research \cite{starck2002curvelet,do2005contourlet}, similar continuous domain representations have attracted less attention. We recently introduced a novel framework recover piecewise polynomial images, whose edges are localized to smooth curves, from their uniform \cite{isbi2015,siam} and non-uniform \cite{sampta2015} Fourier samples; our framework generalizes a recent extension of FRI models to curves \cite{pan2013sampling}. We assume that the partial derivatives of the signal vanish outside the zero level-set of a bandlimited function, which is only true for piecewise smooth signals. This relation translates to a linear system of convolution equations involving the uniform Fourier samples of the partial derivatives, which can be compactly represented as the multiplication of a specific structured matrix with the Fourier coefficients of the bandlimited function. We have introduced theoretical guarantees for the recovery of such images from uniform samples \cite{siam,isbi2015}. Our earlier work has shown that the structured matrix built from the Fourier coefficients of piecewise constant images is low-rank \cite{sampta2015,siam}, which we used to recover the image from its non-uniform Fourier samples with good performance in practical applications. We have also introduced an computationally efficient algorithm termed as GIRAF, which works on the original signal samples rather than the structured high-dimensional matrix \cite{isbi2016,girafarxiv}; the computational complexity of this algorithm is comparable to discrete total variation regularization, which makes this scheme readily applicable to large-scale imaging problems, such as undersampled dynamic magnetic resonance image reconstruction \cite{balachandrasekaran16adm}. 

The main focus of the present paper is to introduce theoretical guarantees on the recovery of continuous domain piecewise constant images from \emph{non-uniform} Fourier samples via a convex structured low-rank matrix completion algorithm. Our main result shows number of non-uniform samples to recovery the image is proportional to the complexity of the edge set, as measured by the bandwidth of the edge set function, and an incoherence measure related to the edge set geometry. We additionally show that the recovery is robust to noise and model-mismatch.

The proof of the main result builds off of \cite{chen2014robust}, which proved similar recovery guarantees for the recovery of multi-dimensional isolated Diracs from non-uniform Fourier samples by minimizing the nuclear norm of an ``enhanced'' multi-level Hankel matrix. This work showed that the number of samples necessary for recovery depends the number of Diracs and on an incoherence measure of the signal, that can be defined solely in terms of the relative locations of the Diracs. 
However, the theory in \cite{chen2014robust} relies heavily on an explicit factorization of the enhanced matrix (e.g., Vandermonde factorization of a Hankel matrix in the 1-D case), which is only available when the number of singularities are isolated and finite. Since the singularities in the proposed class of piecewise constant images (i.e., the image edges) are not isolated nor finite, the recovery guarantees in \cite{chen2014robust} cannot be directly extended to our setting.
Instead, to achieve our result, we give a new characterization of the row and column spaces of the structured matrix arising in our setting.
We show this new characterization allows us to derive an incoherence measure based solely on geometric properties of the edge set. In particular, we derive an upper bound for the incoherence measure that is related to the size of edge set curve. The results show that high sampling burden is associated with the estimation of images with smaller piecewise constant regions, which is consistent with intuition. 

We note that the signal models in \cite{chen2014robust,bhaskar2011atomic,candes2014towards} do not include the class of piecewise constant images considered in this work. In particular, all of the above models assume the discontinuities to be finite in number and well separated, unlike in our setting. Recently, \cite{ye2016compressive} adapted the results in \cite{chen2014robust} to introduce recovery guarantees for Fourier interpolation of a variety of finite-rate-of-innovation signal models \cite{vetterli2002sampling}, including piecewise constant functions. However, these results are limited to the 1-D setting and share the assumption than the discontinuities/innovations of the signal are finite and isolated. Furthermore, the structured matrix lifting considered in this work is different than those considered in \cite{chen2014robust} and \cite{ye2016compressive}. Specifically, the structured matrix lifting in this work consists of two vertically concatenated multi-level Toeplitz matrices (i.e., block Toeplitz with Toeplitz blocks), whose entries are built from the weighted Fourier coefficients of the images. This is substantially different from the structured matrix liftings considered in \cite{chen2014robust} (unweighted, one block, single block multilevel Hankel) and \cite{ye2016compressive} (weighted, one block, single-level Hankel). Finally, we note that a preliminary version of the results presented in this has been published previously in the conference paper \cite{icip2016} without proofs. The present work includes considerably more details and proofs, and major improvements to the main theorem.

\subsection{Notation}
Bold lower-case letters $\mbf x$ are used to indicate vector quantities, bold upper-case $\mbf X$ to denote matrices, and calligraphic script $\mathcal{X}$ for general linear operators. We typically reserve lower-case greek letters $\mu, \gamma$, \emph{etc.}\ for trigonometric polynomials \eqref{eq:trigpoly} and upper-case greek letters $\Lambda, \Omega,$ \emph{etc}.\ for their coefficient index sets, i.e.\ finite subsets of the integer lattice $\mathbb{Z}^2$, with cardinality denoted by $|\Lambda|$. We write $\Lambda + \Omega$ for the dilation of the index set $\Omega$ by $\Lambda$, i.e.\ the Minkowski sum $\{\mbf k + \mbf \ell : \mbf k\in\Lambda,~\bs\ell \in \Omega\}$, and write $2\Lambda$ to mean $\Lambda+\Lambda$, $3\Lambda = 2\Lambda +\Lambda$, \emph{etc}. We also denote the contraction of $\Omega$ by $\Lambda$ by $\Omega\,{:}\,\Lambda = \{\mbf \ell \in \Omega\,{:}\, \mbf \ell - \mbf k \in \Omega \text{ for all } \mbf k \in \Lambda\}$.
\section{Background}
\begin{figure}
\centering
\includegraphics[width=0.45\textwidth]{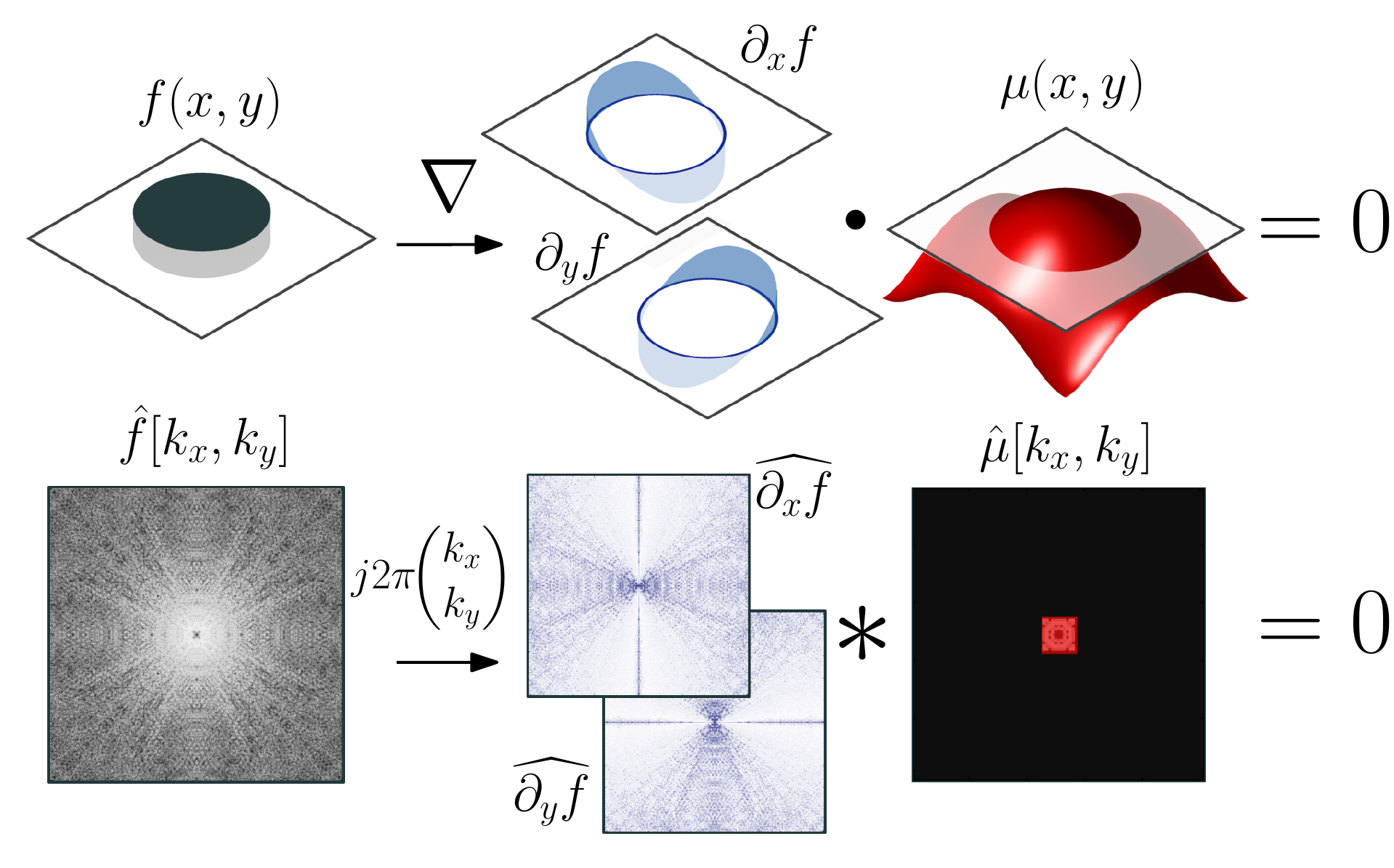}
\caption{Annihilation of a piecewise constant function as a multiplication in spatial domain (top) and as a convolution in Fourier domain (bottom). The partial derivatives of a piecewise constant function are supported on the edge set. If there is a bandlimited function $\mu$ that is zero along the edge set, then the spatial domain product of $\mu$ with the gradient $\nabla f = (\partial_x f,\partial_y f)$ is identically zero. In Fourier domain, this is equivalent to the annihilation of the arrays $j2\pi k_xf[k_x,k_y]$ and $j2\pi k_y f[k_x,k_y]$ by 2-D convolution with a finite filter determined by the Fourier coefficients $\hat \mu$.}
\label{illus}
\end{figure}
\subsection{2-D Piecewise Constant Images with Bandlimited Edges}
In this work we consider a continuous domain \emph{piecewise constant} model for images, 
\begin{equation}
\label{eq:pwc}
f(\bm r) = \sum_{i=1}^N a_i ~1_{U_i}(\bm r),~~\text{ for all }~~\bm r=(x,y) \in{[0,1]}^2,
\end{equation}
where $a_i \in \mathbb{C}$, $1_{U}$ denotes the characteristic function of the set $U$, and each $U_i \subset [0,1]^2$ is a simply connected regions with piecewise smooth boundaries $\partial U_i$. We study the recovery of such an image from a sampling of its Fourier coefficients $\hat f$ specified by
\begin{equation}
\hat f[\bm k] = \int_{[0,1]^2} f(\bm r) e^{-j 2 \pi \bm k\cdot \bm r}; ~\bm k \in \Omega \subset \mathbb{Z}^2.
\end{equation}
Following \cite{siam}, we further assume that the edge set of the piecewise constant image, specified by $E:=\cup_i \partial U_i$, coincides with the zero set of a 2-D bandlimited function:
\begin{equation}
E = \{\bm r \in [0,1]^2: \mu(\bm r) = 0\},~\text{with}~
\mu(\bm r) = \sum_{\bm k\in{ \Lambda}} c[\bm k]\, e^{j2\pi\bm k \cdot \bm r},
\label{eq:trigpoly}
\end{equation}
where the coefficients $c[\bm k]\in\mathbb{C}$, and ${\Lambda}$ is a finite subset of $\mathbb{Z}^2$. We call any function $\mu$ in the form \eqref{eq:trigpoly} a \emph{trigonometric polynomial}, and we say $\mu$ is bandlimited to $\Lambda$, i.e., the Fourier coefficients $\hat \mu$ are supported within $\Lambda$. For short, we will write $\{\mu= 0\}$ for the zero set of $\mu$ considered as a subset of $[0,1]^2$.

Define the \emph{degree} of a trigonometric polynomial $\mu$, denoted by $deg(\mu) = (K,L)$ to be the linear dimensions of the smallest rectangle containing the support set $\{\bm k : \hat\mu[\bm k]\neq 0\}$. In \cite{siam} we proved that for every curve $E$ given by the zero set of a trigonometric polynomial, there exists a unique minimal degree trigonometric polynomial\footnote{More precisely, $\mu_0$ is unique up to multiplication by a phase factor $e^{j2\pi \bm k \cdot \bm r}$ for some $\bm k \in \mathbb{Z}^2$.} $\mu_0$ such that $E = \{\mu_0=0\}$ and if $\mu$ is any other trigonometric polynomial with $\{\mu_0=0\} \subset \{\mu = 0\}$, then $deg(\mu_0) \leq deg(\mu)$ entrywise. By extension, we define the \emph{degree} of a curve $E$ to be equal to the degree of of its minimial degree polynomial $\mu_0$. We also say the curve $E$ is \emph{bandlimited} to $\Lambda_0\subset\mathbb{Z}^2$, where $\Lambda_0$ is the minimal rectangular index set containing the support of $\hat \mu$. Intuitively, the degree/bandwidth of a curve gives a quantitative measure of its complexity. For example, in \cite{siam} we show the number of connected components of a curve is bounded in terms of its degree.
\subsection{Recovery from uniform Fourier samples}
We have shown in \cite{siam} that when $\mu$ is any trigonometric polynomial that vanishes on the edge set of the piecewise constant image $f$, the gradient $\boldsymbol\nabla f = \left(\partial_x f, \partial_y f\right)$ satisfies the property 
\begin{equation}
\mu \boldsymbol\nabla f = 0,
\label{eq:spacedom2}
\end{equation}
where equality in \eqref{eq:spacedom2} is understood in the sense of distributions (see, e.g., \cite{strichartz03}). 
The spatial domain annihilation relation \eqref{eq:spacedom2} translates directly to the following convolution annihilation relation in Fourier domain:
\begin{equation}
\sum_{\bm k \in {\Lambda}} \widehat{\boldsymbol\nabla f}[\bm\ell - \bm k]\;  \widehat{\mu}[\bm k]= \bm 0, 
   ~~\forall ~\bm \ell\in \mathbb{Z}^2.
   \label{eq:annsys}
\end{equation}
Here $\widehat{\boldsymbol\nabla f}[\bm k] = j2\pi(k_x \widehat{f}[\bm k],k_y \widehat{f}[\bm k])$ for $\bm k=(k_x,k_y)$. Note the equations in \eqref{eq:annsys} are linear with respect to the coefficients $\widehat\mu$.

Suppose we have access to samples of the Fourier coefficients $\hat f$ on a finite rectangular grid $\Gamma \subset \mathbb{Z}^2$, and suppose $\mu$ is bandlimited to $\Lambda_1 \subset \Gamma$. Then we can build the system of equations in \eqref{eq:annsys} for all $\bm \ell$ belonging to the index set $\Lambda_2\subset \Gamma$, where $\Lambda_2$ is the set of all integer shifts of $\Lambda_1$ contained in $\Gamma$. In this case \eqref{eq:annsys} can be compactly represented in matrix form as
\begin{equation}
\Tf \bm h=
\begin{bmatrix}
\Txf\\
\Tyf
\end{bmatrix} \bm h = \bm 0 ,
\label{eq:annmatrix}
\end{equation}
where $\Txf,\Tyf \in \mathbb C^{|\Lambda_2| \times |\Lambda_1|}$ are matrices corresponding to the discrete 2-D convolution with the arrays $k_x \widehat f[k_x,k_y]$ and $k_x \widehat f[k_x,k_y]$ for $(k_x,k_y) \in \Gamma$, respectively (after omitting the inconsequential factor $j2\pi$). Here we use $\bm h$ to denote the vectorized version of the filter $(\hat \mu[\bm k]: \bm k \in \Lambda_1)$, where the index set $\Lambda_1$ is called the \emph{filter support}. The matrices $\Txf$ and $\Tyf$ have a block Toeplitz with Toeplitz blocks structure. See Figure \ref{fig:lifting} for an illustration of the construction of $\Tf$. 

\begin{figure}[ht!]
\centering
\includegraphics[width=\columnwidth]{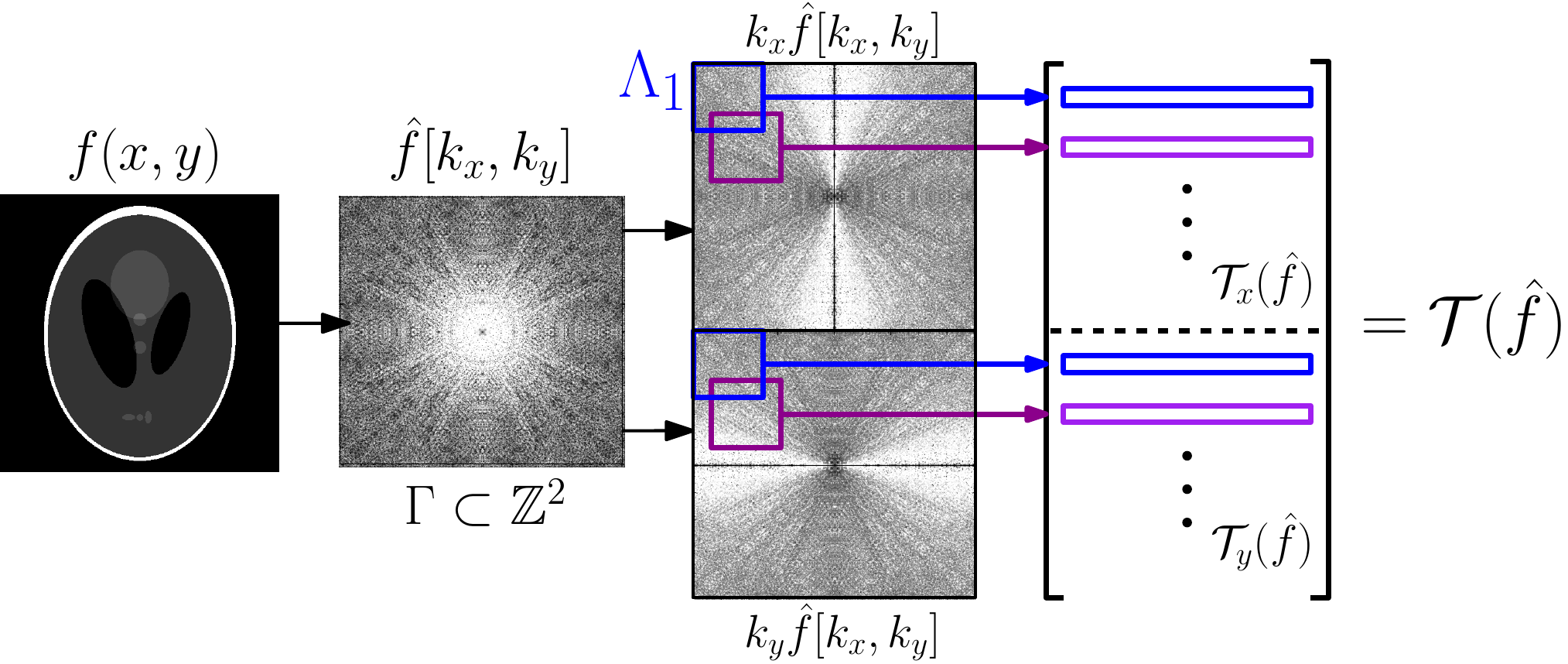}
\caption{Construction of the structured matrix lifting $\Tf$ considered in this work. From a rectangular array of the Fourier coefficients $\hat f [k_x,k_y]$ of a continuous domain image $f(x,y)$, the weighted arrays $k_x\hat f [k_x,k_y]$ and $k_y\hat f [k_x,k_y]$ are constructed. The matrices $\Txf$ and $\Tyf$ are then obtained by extracting all vectorized patches from the weighted arrays, and loading these into the rows of $\Txf$ and $\Tyf$. The resulting matrices $\Txf$ and $\Tyf$ have a block Toeplitz with Toeplitz blocks structure. Finally $\Tf$ is formed by vertically concatenating the blocks $\Txf$ and $\Tyf$.}
\label{fig:lifting}
\end{figure}

Equation \eqref{eq:annmatrix} shows that $\Tf$ is rank deficient, since it has the non-trivial vector $\bm h$ in its nullspace.
In addition, when the filter support $\Lambda_1$ defining $\Tf$ is sufficiently big, we can also show $\Tf$ is low-rank. This is because if $\mu_0$ is the minimal degree polynomial for the edge set, then \emph{any} multiple of $\mu = \gamma\cdot \mu_0$ bandlimited to $\Lambda_1$ will satisfy the annihilation equation \eqref{eq:spacedom2}. In Fourier domain, this means the vector 
\begin{equation}\label{eq:nullspace}
\bm h = ((\hat\mu_0 \ast \hat \gamma)[\bm k] : \bm k\in\Lambda_1) 
\end{equation}
is in the nullspace of $\Tf$. Hence if the filter support $\Lambda_1$ is larger than support $\Lambda_0$ of $\mu_0$, $\Tf$ has a large nullspace and is low-rank. The following result from \cite{siam} gives an exact characterization of the rank of $\Tf$, which will be important for this work:


\begin{thm}\label{prop:rank}\cite{siam} Suppose $f$ is a piecewise constant image \eqref{eq:pwc} whose edge set $E=\{\mu_0 = 0\}$ is the zero set of a trigonometric polynomial $\mu_0$ bandlimited to $\Lambda_0$. Let $\Tf$ be built with filter size $\Lambda_1\supseteq \Lambda_0$, then
\begin{equation}\label{eq:rank}
\text{rank}~\Tf \leq |\Lambda_1|-|\Lambda_1\colon\Lambda_0|
\end{equation}
where $|\Lambda_1|$ is the number of indices in $\Lambda_1$ and $|\Lambda_1\colon\Lambda_0|$ is the number of integer shifts of $\Lambda_0$ contained in $\Lambda_1$. Moreover, equality holds in \eqref{eq:rank} if $\Gamma \supseteq 2\Lambda_1 + \Lambda_0$ and if the edge set does not contain any singular points. In this case, the nullspace of $\Tf$ consists of all vectors in the form \eqref{eq:nullspace}.
\end{thm}

Note that $R:=|\Lambda_1|-|\Lambda_1\colon\Lambda_0|$ is a measure of the bandwidth of $\mu_0$ and hence is indicative of the complexity of the edge set curve $E = \{\mu_0=0\}$. In the remainder of this work we assume the conditions in Theorem \ref{prop:rank} that guarantee the equality $\text{rank}~\Tf = R$ holds, in particular $\Gamma \supseteq 2\Lambda_1 + \Lambda_0$.

If we take $\Lambda_1 = \Lambda_0$, the above result shows Fourier samples of $\hat f$ in $\Gamma \supseteq 3\Lambda_0$ is sufficient for the recovery of the minimal degree polynomial $\mu_0$, since in this case $\hat{\mu}_0$ can be identified as the unique non-trivial nullspace vector of $\Tf$. The following theorem states that once $\mu_0$ is available, $f$ is the unique solution to the annihilation equations \eqref{eq:spacedom2} and \eqref{eq:annsys}:
\begin{thm}\cite{siam}.
\label{thm:uniqueamplitudes}
Suppose $f$ is a piecewise constant image \eqref{eq:pwc} whose edge set $E=\{\mu_0 = 0\}$ is the zero set of a trigonometric polynomial $\mu_0$ bandlimited to $\Lambda_0$. Suppose the Fourier sampling set $\Gamma \supseteq \Lambda_0$. If $g \in L^1([0,1]^2)$ satisfies
\begin{equation}
\mu_0\, \nabla g = 0~\text{subject to}~\hat{g}[\mbf k] = \hat{f}[\mbf k]~\text{for all}~\mbf k \in {\Gamma},
\label{eq:infiniterecovery}
\end{equation}
then $g = f$ almost everywhere.
\end{thm}
In principle, this result allows us to solve for the amplitudes of regions of the piecewise constant function $f$ by plugging in the known $\mu_0$ into the equation \eqref{eq:infiniterecovery} and solving a linear system, similar to Prony's method. However, for complicated piecewise constant images with many regions, it may be more practical to use the approximations introduced in \cite{siam}.

\section{Recovery from non-uniform Fourier samples}
The theory presented in Section \ref{sec:intro} shows that the exact recovery of a continuous domain piecewise constant image with a bandlimited edge set is possible when we collect Fourier samples of the image on a sufficiently large uniform grid in Fourier domain. However, the recovery procedure breaks down when we have non-uniform or missing samples, which is often the case in practical settings, e.g., compressed sensing MRI \cite{lustig2007sparse}. Therefore, we propose and analyze a method to interpolate the missing samples to a uniform grid in Fourier domain, which guarantees full recovery of the image in spatial domain.

Recall that Theorem \ref{prop:rank} says that the structured matrix $\Tf$ built from the Fourier coefficients $\hat f[\bm k],\bm k \in \Gamma$, where $\Gamma \subset \mathbb{Z}^2$ is a uniform rectangular grid, is known to be low-rank precisely when $f$ is a piecewise constant image with a bandlimited edge set. Hence we propose to recover $\hat f[\bm k],\bm k \in \Gamma$ from its samples at non-uniform locations $\Omega \subset \Gamma$ as the solution to the convex matrix completion problem:
\begin{equation}
\label{nucnorm}
\min_{\hat g[\bm k], \bm k\in\Gamma}~ \|\mathcal{T}(\hat{g})\|_* ~\mbox{subject to}~ \hat{g}[\bm k] =  \hat{f}[\bm k]~\text{for all}~\bm k\in \Omega
\end{equation}
where $\|\cdot\|_*$ denotes the nuclear norm, i.e., the sum of the singular values of a matrix, which is the convex relation of the rank functional. Note that \eqref{nucnorm} is different than the standard low-rank matrix completion setting studied in \cite{Candes2012,gross} in that the low-rank matrix $\Tf$ is structured and parameterized by the coefficient vector $\hat f$. Similar structured low-rank matrix completion schemes have been proposed for the recovery of signals from non-uniform Fourier samples \cite{chen2014robust,ye2016compressive} and used with empirical success in MRI applications \cite{sake,haldar2014low,sampta2015}.
The main focus of this paper is to determine the sufficient number of samples that will ensure exact recovery of the Fourier coefficients of $f$ on the reconstruction grid $\Gamma$ with high probability.



\subsection{Role of incoherence}

Several authors have shown that the sufficient number of samples for low-rank matrix recovery by nuclear norm minimization to succeed is dependent on the \emph{incoherence} of the sampling basis with respect to the matrix to be to be recovered \cite{gross,chen2014robust}. Similarly, our results depend on an incoherence measure derived from the structure of the matrix $\Tf$ and properties of the piecewise constant image $f$. In particular, define $\mathcal P_U$ and $\mathcal P_V$ to be the orthogonal projections onto the column space and row space of $\Tf$, respectively, i.e.,\ if $\Tf = \bm U\bm \Sigma\bm V^*$ is the rank-$R$ singular value decomposition then $\mathcal P_U \bm X = \bm U\bm U^* \bm X$, $\mathcal P_V \bm X = \bm X \bm V\bm V^*$. In Appendix B, we show that the structured matrix $\Tf$ can be expanded using orthonormal basis of matrices $\mathbf A_k$ such that
 \begin{equation}
 	\Tf = \sum_{\bm k\in\Gamma/\{0\}} \hat f[\bm k] w[\bm k] \bm A_{\bm k}
 \end{equation}
 where $w[\bm k],\bm k\in\Gamma/\{0\}$ are a set of positive weights that do not depend on $\hat f$. Similar to results in \cite{gross,chen2014robust,ye2016compressive}, we prove that nuclear norm minimization \eqref{nucnorm} recovers the exact low-rank solution with high probability provided we can uniformly bound the norms of the projections of the sampling basis matrices $\bm A_{\bm k}$ onto the row and column spaces of $\Tf$:

\begin{prop}\label{prop:coherence}
Consider $\Tf$ of rank $R$ corresponding to a piecewise constant function $f$ whose edge set coincides with the zero set of $\mu_0$, let $\rho$ be the incoherency measure associated with $\mu_0$ to be defined in the sequel, and set $c_s = |\Gamma|/|\Lambda_1|$. Then we have
\begin{align}
\label{coherence2}
\max_{\bm k \in \Gamma} ~\|\mathcal P_{U}\mathbf A_{\bm k}\|_F^2 &\leq \frac{\rho\,R\,c_s}{|\Gamma|},\\
\max_{\bm k \in \Gamma}\|\mathcal P_{V}\mathbf A_{\bm k} \|_F^2 &\leq  \frac{\rho\,R\,c_s}{|\Gamma|}
\end{align}
\end{prop}
The proof in Section \ref{coherenceappendix} relies on the row and column spaces of $\Tf$ derived in Lemma \ref{columnlemma} and Lemma \ref{rowlemma} in the next section. These results will be used in the derivation of the main theorem in Section \ref{mainthmproof}.

\subsection{Main Results}
We now present our main results, which determine the sufficient number of random Fourier samples for the convex structured low-rank matrix completion program \eqref{nucnorm} to succeed with high probability. 
Our first theorem addresses the case of recovery from noiseless Fourier samples:
\begin{thm}
\label{mainthm}
Let $f$ be a continuous domain piecewise constant image \eqref{eq:pwc}, whose edge-set is described by the zero-set of the trigonometric polynomial $\mu_0$ bandlimited to $\Lambda_0$ (see \eqref{eq:trigpoly}). Let $\Omega\subset \Gamma$ be an index set drawn uniformly at random within $\Gamma$. Then there exists a universal constant $c>0$ such that the solution to \eqref{nucnorm} is $\widehat{f}$ with probability exceeding $1-|\Gamma|^{-2}$, provided 
\begin{equation}\label{eq:numsamples}
|\Omega| > c \,\rho\;c_s\;R\;\log^{4}|\Gamma|,
\end{equation}
where $R = |\Lambda_1|-|\Lambda_1:\Lambda_0| = \text{rank}~\Tf$, $c_s = |\Gamma|/|\Lambda_1|$, $c$ is a universal constant, and $\rho \geq 1$ is an incoherence measure depending on the geometry of the edge-set, to be defined in the sequel.
\end{thm}

To better understand the dependence of the bound in \eqref{eq:numsamples} on the filter size $\Lambda_1$ and the edge set bandwidth $\Lambda_0$, assume for simplicity that $\Lambda_1$ is some dilation of $\Lambda_0$, that is, $\Lambda_1 = \alpha \Lambda_0$, where $\alpha > 1$ is an integer. In this case, the factor $c_s\,R$ in 
\eqref{eq:numsamples} simplifies to 
\begin{equation}\label{eq:boundsimp}
\left(\frac{|\Lambda_1|-|\Lambda_1:\Lambda_0|}{|\Lambda_1|}\right)|\Gamma|\leq \left(\frac{\alpha^2-(\alpha-1)^2}{\alpha^2}\right)|\Gamma|\leq \frac{2|\Gamma|}{\alpha}.
\end{equation}
Therefore, assuming the other constants in \eqref{eq:numsamples} are fixed, the number of measurements sufficient for exact recovery is proportional to the reciprocal of the dilation factor $\alpha$. This suggests taking the filter size $\Lambda_1$ to be as large as allowed by Theorem \ref{mainthm}. Namely, $\Lambda_1$ should satisfy $2\Lambda_1 + \Lambda_0 = \Gamma$, i.e., the side-lengths of filter support $\Lambda_1$ should be roughly half those of the reconstruction grid $\Gamma$. Fixing the filter support $\Lambda_1$ to obey this bound, then $\Gamma = (2\alpha+1)\Lambda_0$, and so $|\Gamma|\leq (2\alpha+1)^2|\Lambda_0|$. Inserting this bound into \eqref{eq:boundsimp} gives 
\begin{equation}
c_s R = O\left(\alpha|\Lambda_0|\right).
\end{equation}
Combined with \eqref{eq:numsamples}, this shows that the number of measurements sufficient for exact recovery is on the order of $|\Lambda_0|$, up to incoherence and log factors.


The proof of Theorem \ref{mainthm}, detailed in Appendix B, is in line with the approach of \cite{chen2014robust}. In particular, we prove the result by constructing an approximate dual certificate using the well-known ``golfing scheme'' of \cite{gross}. The main differences between in the proof of the above result and that in \cite{chen2014robust} results from the differences in the matrix structure and hence the characterization of the incoherency between the row and column subspaces of $\Tf$ with the sampling basis.
In particular, the matrix $\Tf$ we consider is obtained by stacking two block Toeplitz with Toeplitz blocks (BTTB) matrices whose entries are the weighted Fourier coefficients of $f$, as opposed to a single unweighted BTTB matrix in \cite{chen2014robust}. The approach in \cite{chen2014robust} relies on an explicit low-rank factorization of a BTTB matrix in terms of Vandermonde-like matrices\footnote{The structured matrices considered in \cite{chen2014robust} are block Hankel with Hankel block matrices (BHHB), but this difference is purely cosmetic: every BTTB matrix can be re-expressed as BHHB after a permutation of its rows and columns. In particular, the Vandermonde-like factorization of BHHB matrices in \cite{chen2014robust} carries over to BTTB matrices.}. Since this factorization is not available in our setting, we use algebraic properties of trigonometric polynomials to give a new characterization of the row and column spaces of the matrix. In particular, we show in Section \ref{sec:incoherence} that similar Vandermonde-like basis matrices exist for the row and column space of the lifted matrix, and use these to derive a related incoherence measure that satisfies the bounds in Prop.~\ref{prop:coherence}.

\subsection{Recovery in the presence of noise and model-mismatch}
We now generalize \eqref{stlr} to the setting where we have noisy or corrupted Fourier samples 
\begin{equation}
\hat f_n[\bm k] = \hat f[\bm k] + \eta[\bm k], \bm k \in \Omega,
\end{equation}
where $\eta[\bm k] \in \mathbb{C}$ is a vector of noise. In this case, we pose recovery as 
\begin{equation}
\label{eq:noisystlr}
\min_{\hat g}~ \|\Tg\|_* ~\text{subject to}~
\|\mathcal P_{\Omega}(\hat{f_n}-\hat{g})\|_2 \leq \delta.
\end{equation}
where $\delta > 0$ is an estimate of the $\ell^2$-norm of the error $\|\eta\|$, and $\mathcal P_{\Omega}$ denotes projection onto $\Omega$. We make no assumptions on the statistics of the noise $\eta$. In particular, $\eta$ can represent errors due to model-mismatch, such as when the image is not perfectly piecewise constant, or when the edge set of the image does not coincide perfectly with the zero level-set of a bandlimited function.

The following theorem shows that when the deviation of $\hat f_n$ from $\hat f$ is small, the modified recovery program \eqref{eq:noisystlr} recovers a solution that is close in norm to $\hat f$ under the same sampling conditions as Theorem \ref{mainthm}. 

\begin{thm}\label{noisythm}
Let $f$ be specified by \eqref{eq:pwc}, whose edge-set is described by the zero-set of the trigonometric polynomial $\mu_0$ bandlimited to $\Lambda_0$ with associated incoherence measure $\rho$. Let $\Omega\subset \Gamma$ be an index set drawn uniformly at random within $\Gamma$ such that $|\Omega|$ satisfies the bound \eqref{eq:numsamples} in Theorem \ref{mainthm}. If the measurements $\hat f_n$ satisfy $\|\mathcal P_{\Omega}({\hat f_n-\hat f})\|_2 \leq \delta$, then the solution $\hat g$ to \eqref{eq:noisystlr} satisfies 
\begin{equation}\label{eq:noisybound}
\|\Tf-\Tg\|_F \leq 5 |\Gamma|^2 \delta.
\end{equation}
with probability exceeding $1-|\Gamma|^{-2}$.
\end{thm}
See Section \ref{noisythmproof} in the Supplementary Materials for proof. The bound \eqref{eq:noisybound} allows us to quantify the effect of model-mismatch on recovery. In particular, suppose the image $f_n$ represents a perturbation from an ideal piecewise constant image $f$ such that their difference in $L^2$-norm is $\delta$-small:
\begin{equation}
\|f_n-f\|_{L^2}^2 = \left(\int_{[0,1]^2}|f_n(\bm r)-f(\bm r)|^2\,d\bm r\right)^{\frac 1 2} \leq \delta.
\end{equation}
Then by Parseval's theorem, the measurements of $\hat f_n$ satisfy $\|\mathcal P_{\Omega}({\hat f_n-\hat f})\|^2 \leq \delta$, hence Theorem \ref{noisythm} applies. From \eqref{eq:noisybound} we obtain the bound
$\|\Tf-\Tg\|_F \leq 5 |\Gamma|^2 \|f_n-f\|_{L^2}.$
This shows that if the image $f_n$ is close to the ideal piecewise constant image $f$ in spatial domain $L^2$-norm, then the matrix $\Tg$ we recover using \eqref{eq:noisystlr} will be close in norm to $\Tf$ with high probability.

\section{Row and column spaces of $\Tf$ and incoherence}\label{sec:incoherence}
In this section we define an incoherence measure $\rho$ that satisfies the desired bounds in Prop.~\ref{prop:coherence}. We show that the incoherence measure depends only on the geometry of the edge set of the image. The incoherence measure is derived from a new characterization of the row and column spaces of the matrix $\Tf$ in terms of Vandermonde-like basis matrices.

\subsection{Row and column spaces of $\Tf$}
Our first lemma gives a basis for the row space of $\Tf$:
\begin{lem} 
\label{rowlemma}
A basis of the row space of $\Tf$ is given by the columns of the $|\Lambda_1|\times R$ Vandermonde-like matrix
\begin{equation}
\label{rowspace}
\ER(P) :=\frac{1}{\sqrt{|\Lambda_1}|} \begin{pmatrix} 
e^{j2\pi \bm k_1 \cdot \bm r_1} & \ldots & e^{j2\pi \bm k_1 \cdot \bm r_R}\\
\vdots &&\vdots\\
e^{j2\pi \bm k_{|\Lambda_1|} \cdot \bm r_1} & \ldots & e^{j2\pi \bm k_{|\Lambda_1|} \cdot \bm r_R}\\
\end{pmatrix}
\end{equation}
where $\{\bm k_1,...,\bm k_{|\Lambda_1|}\}$ is a linear indexing of elements in $\Lambda_1$, and $P =\{\mbf r_1,...,\mbf r_R\}$ is a set of $R=|\Lambda_1|-|\Lambda_1:\Lambda_0|$ distinct points on the edge set curve $\{\mu_0=0\}$ chosen such that the columns of $\ER$ are linearly independent. 
\end{lem}
The careful reader will have noticed that Lemma \ref{rowlemma} takes for granted the existence of a set of points $P =\{\mbf r_1,....,\mbf r_R\} \subset \{\mu_0=0\}$ such that the columns of $\ER(P)$ is linearly independent. Call such a set $P$ a set of \emph{admissible nodes} for the curve $\{\mu_0=0\}$. The following result shows that sets of admissible nodes always exist and are easy to construct:
\begin{lem}\label{lem:admissible}
Let $\mu_0$ be bandlimited to $\Lambda_0$. Any set of $M \geq R + |\Lambda_0|$ distinct points on the curve $\{\mu_0=0\}$ contains a subset of $R$ points that are a set of admissible nodes. 
\end{lem}
The next lemma shows that we can characterize the column space of $\Tf$ in a similar way as the row space:
\begin{lem} 
\label{columnlemma}
A basis of the column space of $\Tf$ is given by the columns of the $2|\Lambda_2| \times R$ weighted Vandermonde-like matrix:
\begin{equation}
\label{columnspace}
\EL(P) = \frac{1}{\sqrt{|\Lambda_2}|}\begin{pmatrix}
\frac{w_{1,x}}{\|\bm w_1\|}\,e^{j2\pi \bm k_1 \cdot \bm r_1} & \hspace{-0.5em}\ldots\hspace{-0.5em} & \frac{w_{R,x}}{\|\bm w_R\|}\,e^{j2\pi \bm k_1 \cdot \bm r_R}\\
\vdots &&\vdots\\
\frac{w_{1,x}}{\|\bm w_1\|}\,e^{j2\pi \bm k_{|\Lambda_2|} \cdot \bm r_1} & \hspace{-0.5em}\ldots\hspace{-0.5em} & \frac{w_{R,x}}{\|\bm w_R\|}\,e^{j2\pi \bm k_{|\Lambda_2|} \cdot \bm r_R}\\
\hline
\frac{w_{1,y}}{\|\bm w_1\|}\,e^{j2\pi \bm k_1 \cdot \bm r_1} & \hspace{-0.5em}\ldots\hspace{-0.5em} & \frac{w_{R,y}}{\|\bm w_R\|}\,e^{j2\pi \bm k_1 \cdot \bm r_R}\\
\vdots &&\vdots\\
\frac{w_{1,y}}{\|\bm w_1\|}\,e^{j2\pi \bm k_{|\Lambda_2|} \cdot \bm r_1} & \hspace{-0.5em}\ldots\hspace{-0.5em} & \frac{w_{R,y}}{\|\bm w_R\|}\,e^{j2\pi \bm k_{|\Lambda_2|} \cdot \bm r_R}\\
\end{pmatrix},
\end{equation}
where where $\{\bm k_1,...,\bm k_{|\Lambda_2|}\}$ is a linear indexing of elements in $\Lambda_2$ and $P = \{\mbf r_1,....,\mbf r_{R}\}$ is a set of admissible nodes for the curve $\{\mu_0 = 0\}$. The weight vectors $\bm w_i = (w_{i,x}, w_{i,y})$, are described by the formula \eqref{aggregated} in Appendix \ref{sec:appendixA}, and depend only on the edge set $\{\mu_0 = 0\}$, the nodes $P$, and the filter support $\Lambda_1$.
\end{lem} 
See Section \ref{rowspaceappendix} for the proofs of Lemmas \ref{rowlemma} and \ref{lem:admissible}, and Section \ref{colspaceappendix} for the proof of Lemma \ref{columnlemma}. 

\subsection{Incoherence measure}
We now show how to define an incoherence measure $\rho$ that satisfies the desired bounds in Prop.~\ref{prop:coherence}. Consider the Gram matrix $\GR(P) = [\ER(P)]^*\ER(P)$, where $P$ is any set of $R$ points $\mbf r_1, ..., \mbf r_R$ on the edge set curve $\{\mu=0\}$. It is easy to see from the definition \eqref{rowspace} that the entries of $\GR(P)$ are specified by
\begin{equation}
(\GR(P))_{i,j} = \frac{1}{|\Lambda_1|}D_{\Lambda_1}(\bm r_i-\bm r_j),~~1\leq i,j \leq R,
\end{equation}
where $D_{\Lambda_1}(\bm r) := \sum_{\bm k \in \Lambda_1} e^{j2\pi \bm k \cdot \bm r}$ is the \emph{Dirichlet kernel} supported on $\Lambda_1$. Note that $\GR(P)$ has ones along the diagonal, and the magnitude of the off-diagonal entries is dictated by the distances $|\bm r_i - \bm r_j|$ and the filter support $\Lambda_1$. We now define the \emph{incoherence measure} $\rho$ associated with the edge set $E=\{\mu_0=0\}$ in terms of $\GR(P)$.

\begin{defn}\label{def:incoherence}
Suppose the edge set curve $E=\{\mu_0 = 0\}$ has bandwidth $\Lambda_0$ (see \eqref{eq:trigpoly}), and set $R = |\Lambda_1|-|\Lambda_1\colon\Lambda_0|$. Define the incoherence measure $\rho$ by
\begin{equation}\label{eq:incoherence}
\rho = \min_{\substack{P \subset \{\mu_0=0\}\\|P|=R}} \frac{1}{\lambda_{min}[\GR(P)]},
\end{equation}
where $\lambda_{min}[\GR(P)]$ is the minimum eigenvalue of $\GR(P)$. 
\end{defn}
Put in words, among all possible arrangements of $R$ points along the edge-set $\{\mu_0=0\}$, we seek the arrangement such that the minimum eigenvalue $\GR(P)$ is as large as possible. Intuitively, the optimal arrangement will maximize the minimum separation distance among the $R$ points, and $\rho$ can be thought of as a measure of this geometric property. In particular, edge set curves that enclose a small area, and hence require the points $P$ to be closely spaced along the curve, will result in a large value of $\rho$. According to Theorem \ref{mainthm}, the measurement burden will be high for such curves.

Note that curves corresponding to a particular bandwidth can come in different sizes. Specifically, for a fixed $\mu_0$ with bandwidth $\Lambda_0$ consider the family of curves $\{\mu_0=\alpha\}$, where $\alpha$ is a scalar. One can change $\alpha$ to obtain multiple curves with exactly the same bandwidth, each of which correspond to a different levelset of $\mu_0$. These level-sets will have different incoherence measures, depending on how large or small the level-set curves are. This shows the incoherence of an edge set captures something besides its bandwidth. See Figure \ref{fig:incoherence} for an illustration.

We can give incoherency measure of an edge set a more precise geometric interpretation based on the minimum separation distance of a set of admissible nodes. We generalize a bound on the condition number of Vandermonde matrices derived in \cite{moitra2015super} to the case of the Vandermonde-like matrix \eqref{rowspace}, and use this to derive a bound for the incoherence parameter $\rho$. 
\begin{thm} 
\label{geometry}
Assume that the points $P =\{( x_i, y_i)\}_{i=1}^R$ belonging to the curve $\{\mu_0=0\}$ satisfy $|x_i-x_j| > \Delta$ and $|y_i-y_j| > \Delta$ for all $i\neq j$. Assume the filter support $\Lambda_1 \subset \mathbb{Z}^2$ is a square region symmetric around the origin of size $\sqrt{|\Lambda_1|}\times \sqrt{|\Lambda_1|}$. Then  
\begin{equation}\label{eq:incoherence_bound}
	\rho \leq \bkt{1-\frac{1}{\sqrt{|\Lambda_1|}~\Delta}}^{-2},
\end{equation}
where $\rho$ is the incoherence parameter \eqref{eq:incoherence} associated with the curve $\{\mu_0 = 0\}$.
\end{thm}
See Section \ref{sec:supp_moitra} of the Supplementary Materials for the proof.
The bound in \eqref{eq:incoherence_bound} shows that the incoherence is close to one (i.e., is as small as possible) when $\Delta \gg 1/\sqrt{|\Lambda_1|}$. Since $\Delta$ is the spacing between each pair of points on the curve, to achieve a larger $\Delta$ spacing, and hence a smaller $\rho$, requires a larger curve. This suggests that fewer measurements are required to recover a larger curve, which is consistent with the findings in the isolated Dirac setting \cite{moitra2015super,ye2015compressive}. 

\begin{figure}[ht!]
\centering
\begin{minipage}{\columnwidth}
\subfloat[][Level-sets of $\mu_0$]{\includegraphics[height=0.23\textwidth]{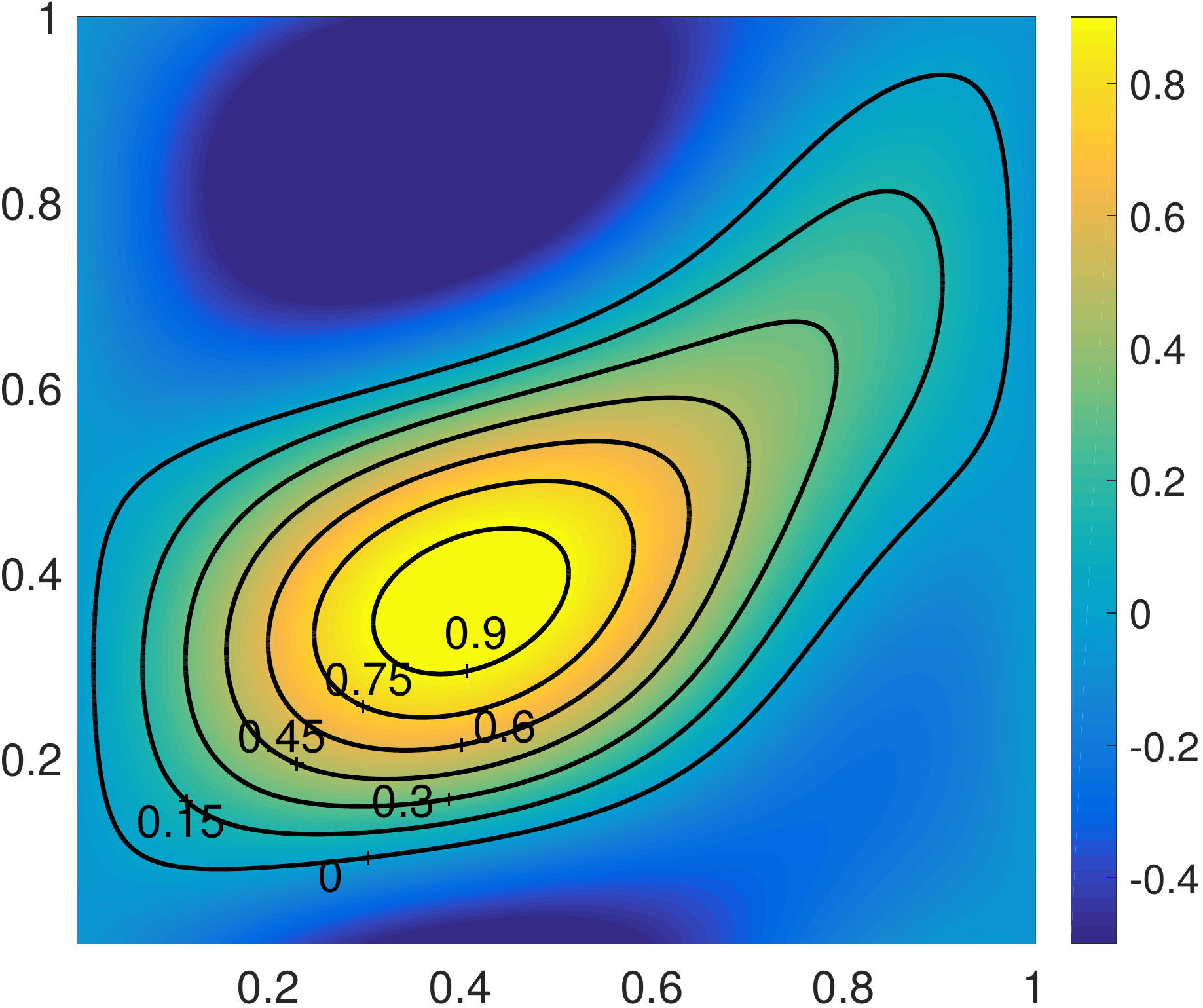}}~
\subfloat[][$\rho \leq 8.0$]{\includegraphics[height=0.23\textwidth]{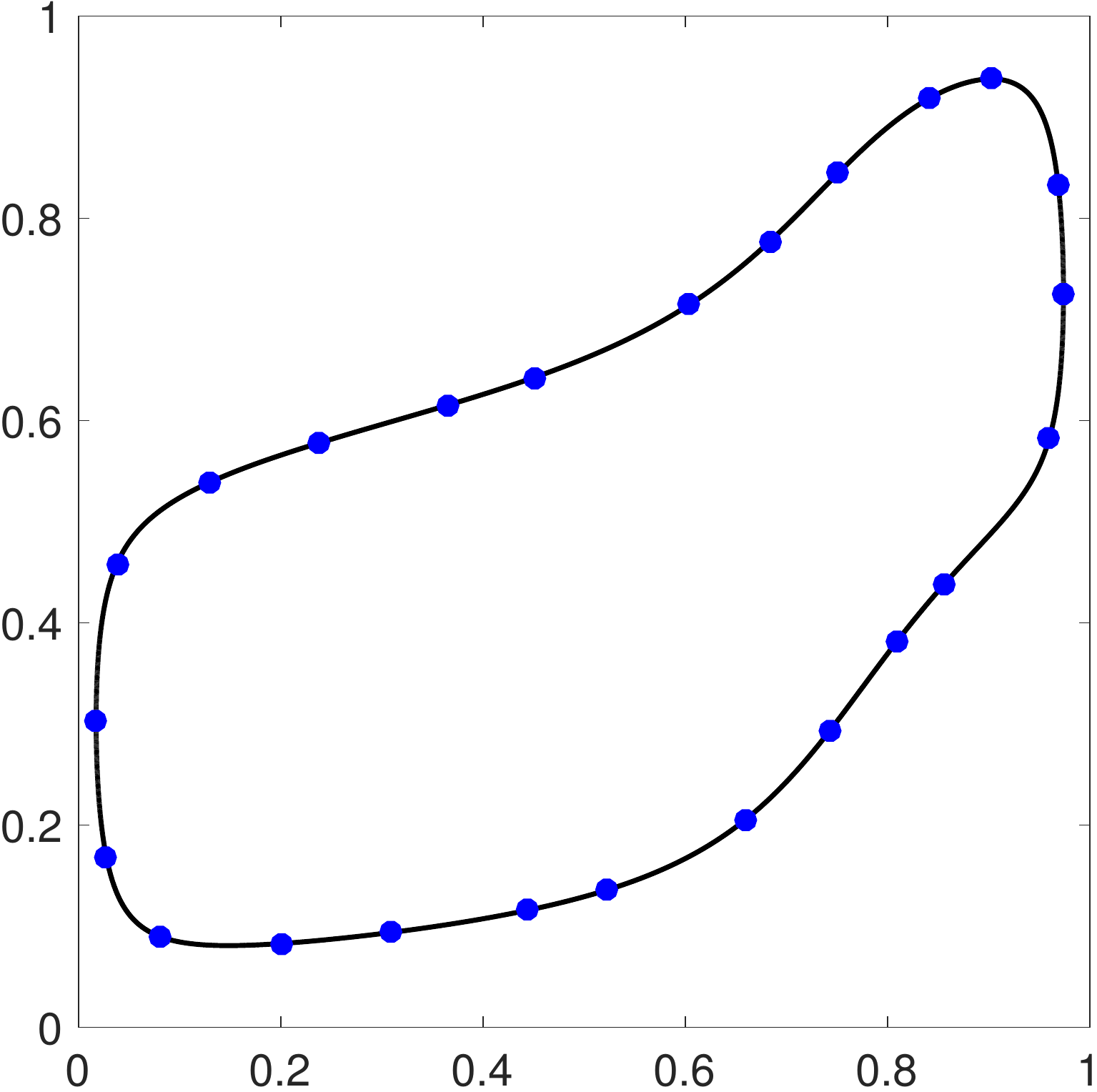}}~
\subfloat[][$\rho \leq 264.9$]{\includegraphics[height=0.23\textwidth]{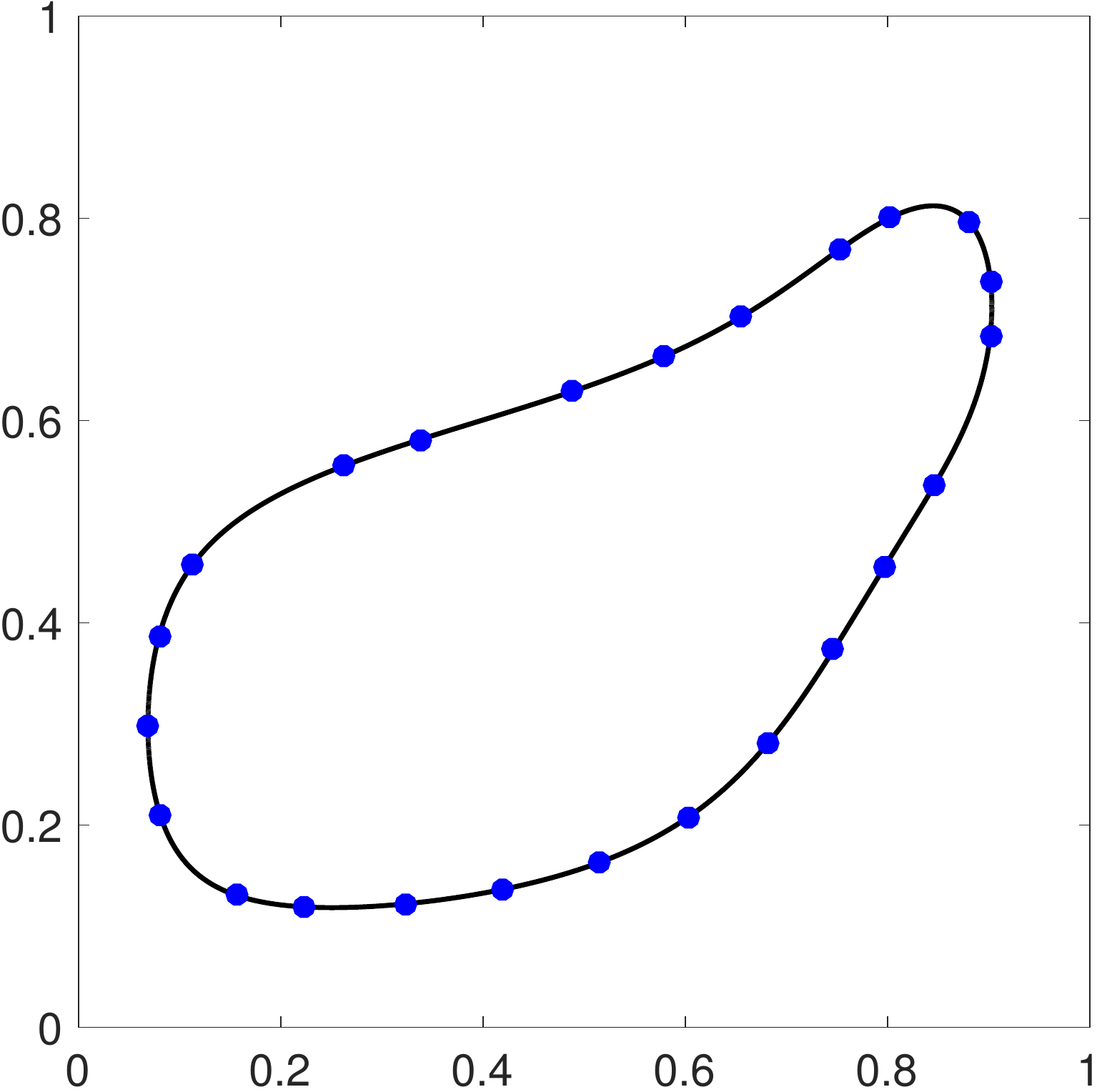}}~
\subfloat[][$\rho \leq 5.0\times 10^4$]{\includegraphics[height=0.23\textwidth]{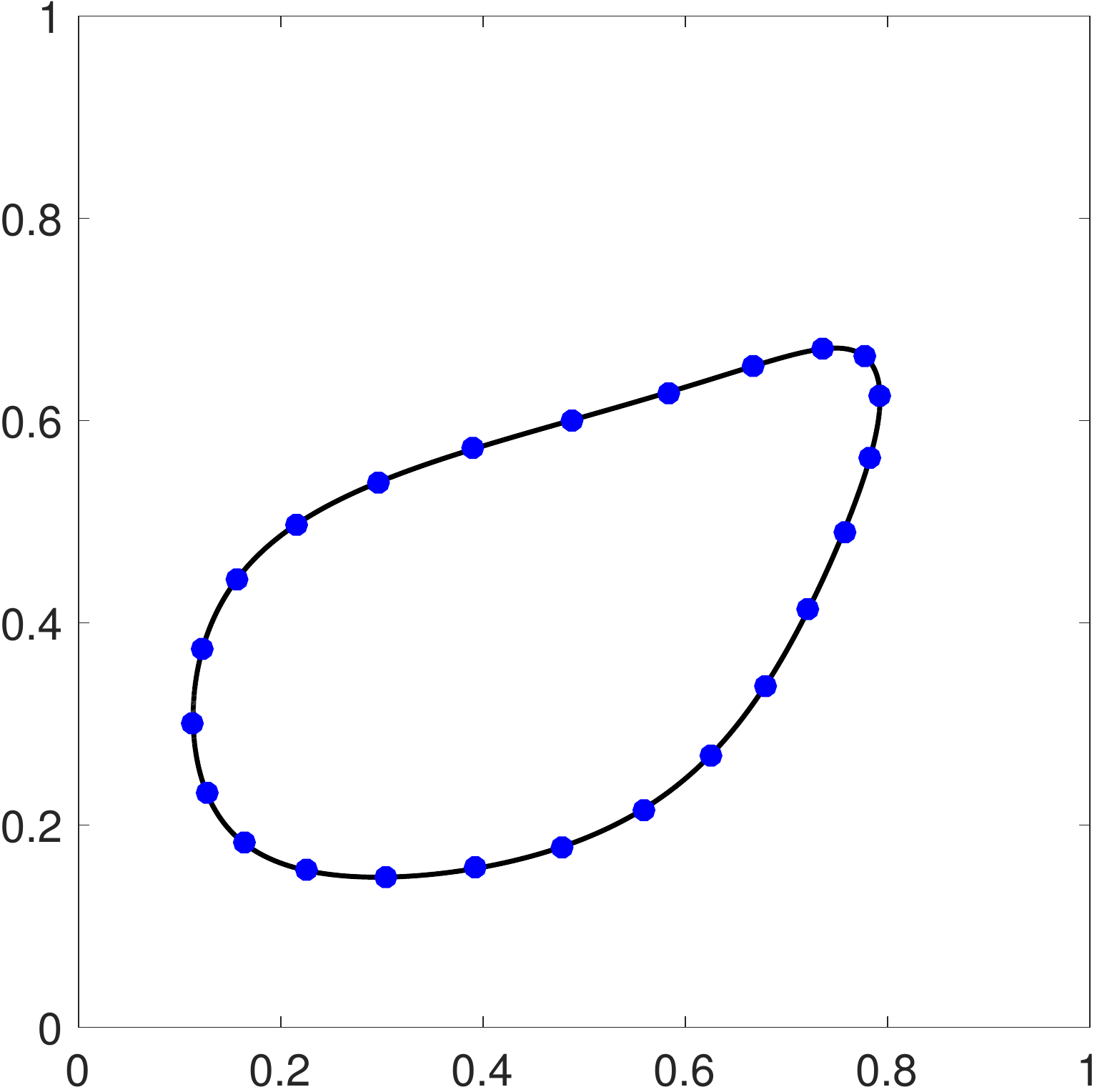}}
\end{minipage}
\caption{\small Illustration of edge set incoherence measure $\rho$. In (a) are the level-sets of trigonometric polynomial $\mu_0$ bandlimited to $\Lambda_0$ of size $3\times 3$. These curves all have the same bandwidth, $\Lambda_0$, but come in different sizes. In (b)-(d) we show $R=24$ nodes on the curve giving the indicated bound on incoherence parameter $\rho$ defined in \eqref{eq:incoherence}, assuming a filter $\Lambda_1$ of size $7\times 7$. Observe that the incoherence measure increases as the curve gets smaller. This indicates the smaller curves have a significant sampling burden.}
\label{fig:incoherence}
\end{figure}

\section{Numerical Experiments}\label{sec:exp}
\subsection{Algorithms}
For small to moderate problem sizes the nuclear norm minimization problem \eqref{nucnorm} can be solved efficiently with the alternating directions method of multipliers (ADMM) algorithm, which results in a modification of the singular value thresholding (SVT) algorithm \cite{cai2010singular}. This approach has been proposed for related structured low-rank matrix completion problems in several works, e.g., \cite{fazel2013hankel,chen2014robust,sampta2015,ye2016compressive}. We adopt this approach here as well for our small-scale numerical experiments. A detailed implementation of this algorithm can be found in, e.g., \cite{ye2015compressive}. However, we note that for large scale problems, such as those encountered in realistic imaging applications, more efficient approaches need to be adopted, because often in these cases the lifted matrix is too large to be held in memory. A fast algorithm for solving an approximation to \eqref{nucnorm} for large-scale problems is given in \cite{girafarxiv}. 
\subsection{Phase transitions}
In Fig.~\ref{fig:phasetransition}, we study the probability of exact recovery under different assumptions on the filter size and edge set of the image. For these experiments the reconstruction grid $\Gamma$ was of size $65\times 65$. We generated synthetic random piecewise constant functions with known edge set bandwidth (see Fig. \ref{fig:incoherence}(c)), and attempted to recover their Fourier coefficients in $\Gamma$ from random samples in $\Omega$ at the specified undersampling factor. For each set of parameters we ran 10 random trials. We count the recovery as ``exact'' if the recovered coefficients $\hat f$ satisfied $\|\hat f -\hat f_0\|/\|\hat f_0\| < 10^{-3}$, where $\hat f_0$ is the ground truth. The exact recovery rate was then obtained by averaging over the 10 trials. 

First, in Fig.~\ref{fig:phasetransition} (a), we studied the effect of changing the filter size $\Lambda_1$ on the recovery while keeping other parameters constant. We fixed the edge-set bandwidth to $|\Lambda_0|=9\times 9$ and varied the filter size as $|\Lambda_1| =(2K+1)\times(2K+1)$ for $K=1,...,30$. We call $K$ the filter bandwidth. Note that Theorem \ref{mainthm} has restrictions on how large $\Lambda_1$ can be. The maximum filter bandwidth for which Theorem \ref{mainthm} holds in this case was $K=15$ (red line in Figure \ref{fig:phasetransition}(a)), however we extended the filter size to observe the behavior of the algorithm outside of this regime. As predicted by Theorem \ref{mainthm}, we find that the optimal performance is obtained when $\Lambda_1$ is the largest as allowed by Theorem \ref{mainthm} (roughly half the size of $\Gamma$ in each dimension).

Next, in Fig.~\ref{fig:phasetransition}(b), we study the recovery as a function of the bandwidth of the edge-set of the image. The filter bandwidth was fixed at $K=15$, and we varied the edge-set bandwidth as $|\Lambda_0| = (2K_0+1)\times(2K_0+1)$. The phase transition shows dependence $|\Omega| = O(|\Lambda_0|)$ as predicted by Theorem \ref{mainthm}.

\begin{figure}[ht!]
\centering
\begin{minipage}{\columnwidth}
\begin{tabular}{cc}
\subfloat[][Varying filter bandwidth]{\includegraphics[height=0.35\textwidth]{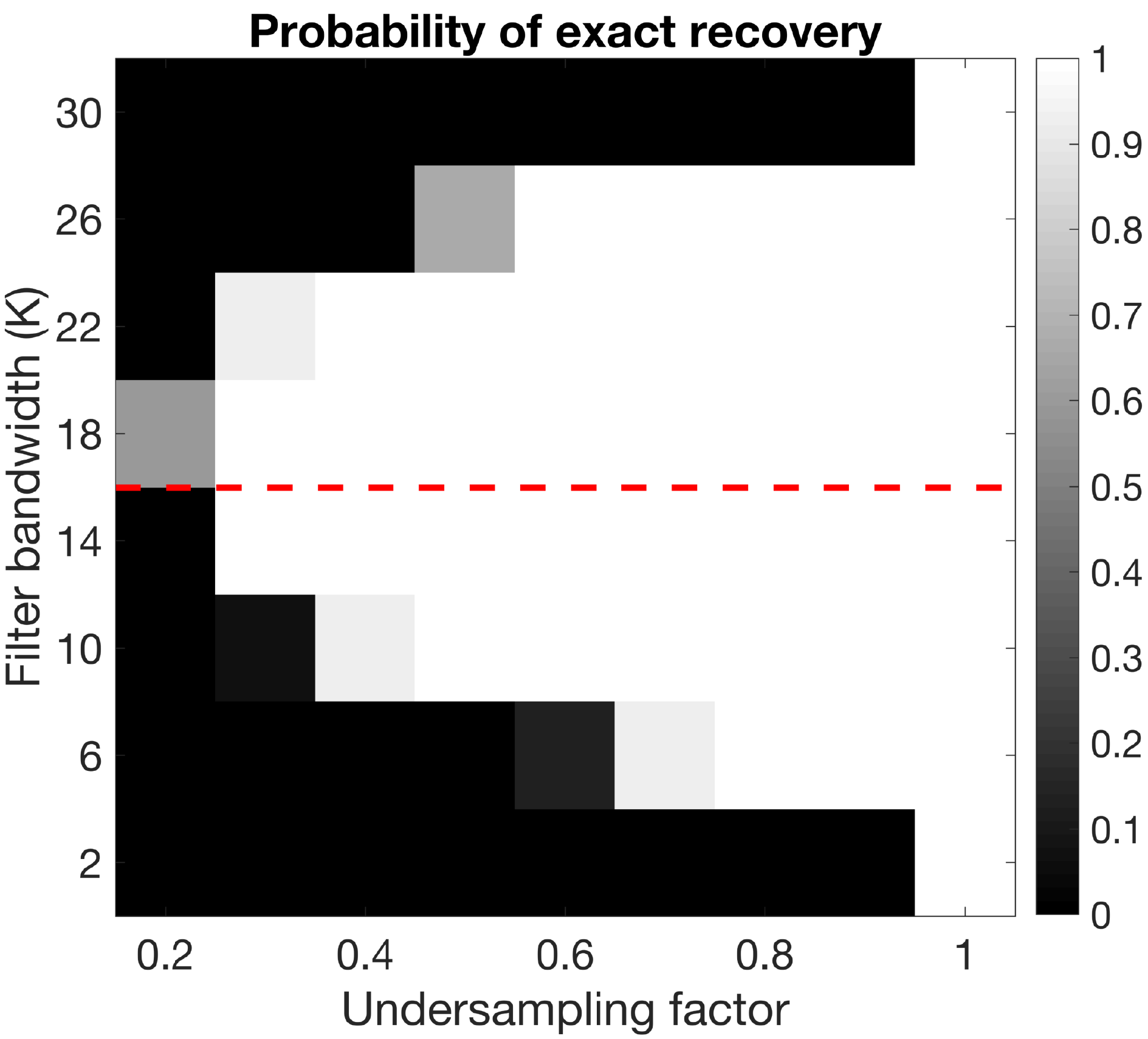}}
& 
\subfloat[][Varying edge set bandwidth]{\includegraphics[height=0.35\textwidth]{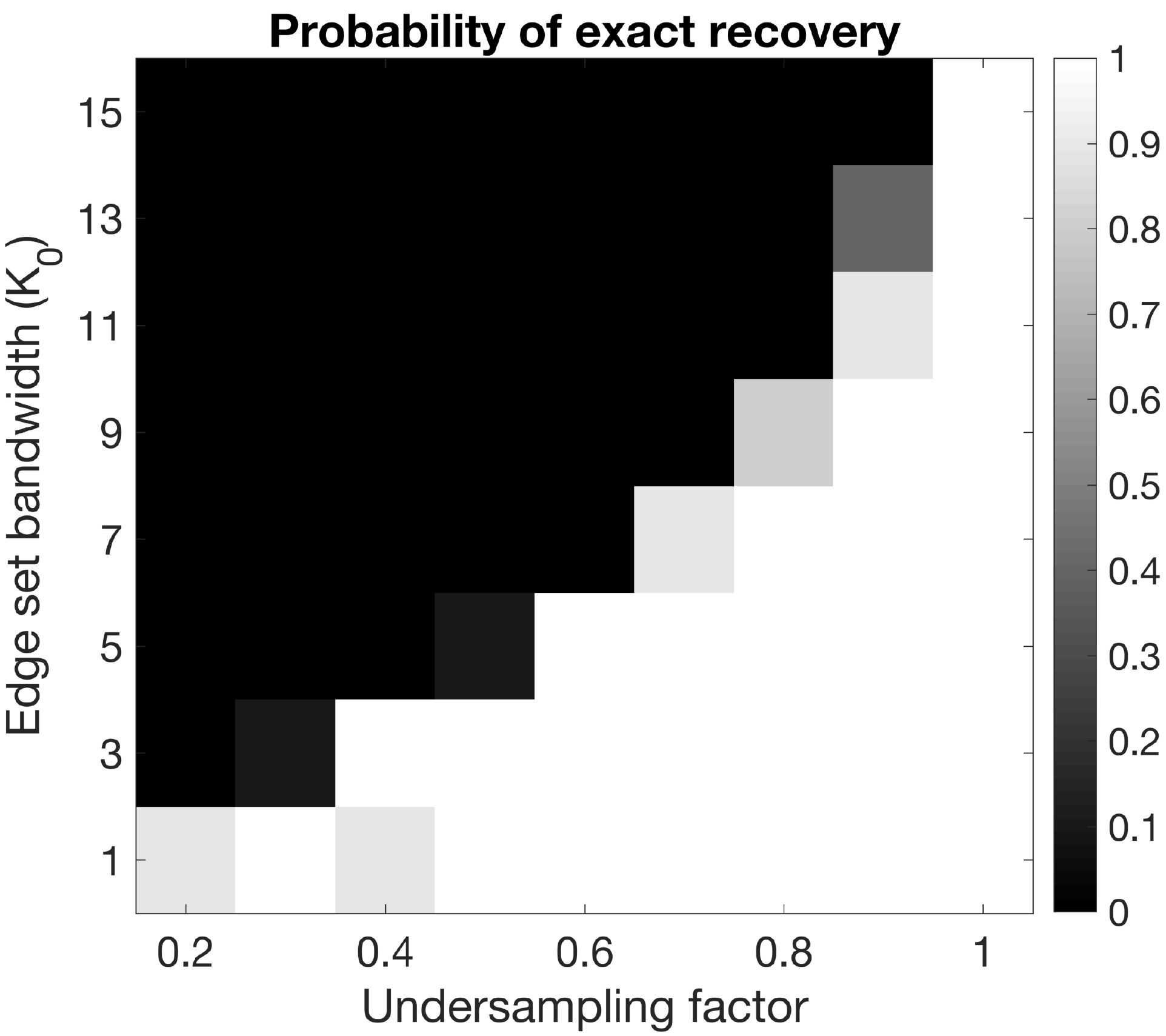}}
\\[2em]
\multicolumn{2}{c}{
\subfloat[][Examples of randomly generated piecewise constant images]{\begin{tabular}{cccc}
\includegraphics[width=0.18\textwidth]{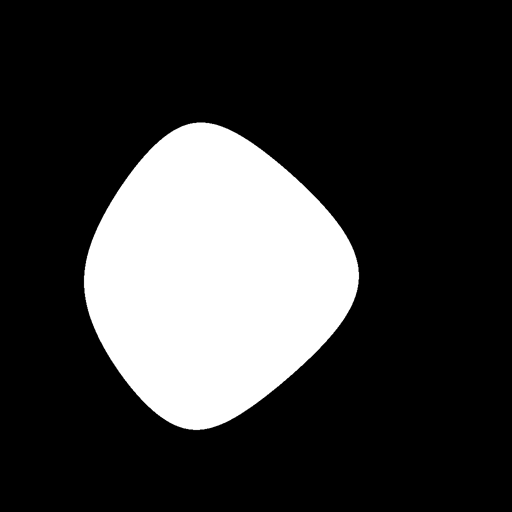} & 
\includegraphics[width=0.18\textwidth]{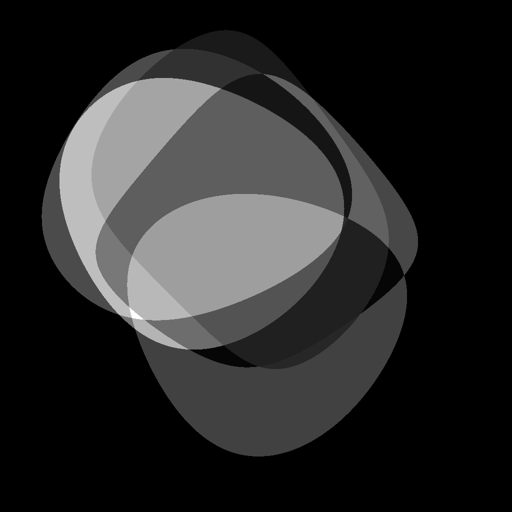} &
\includegraphics[width=0.18\textwidth]{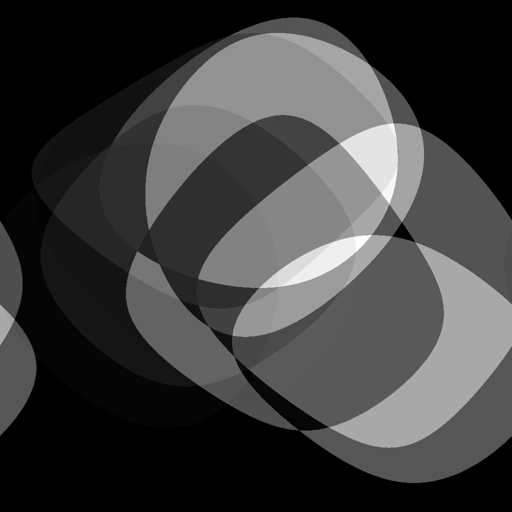} & 
\includegraphics[width=0.18\textwidth]{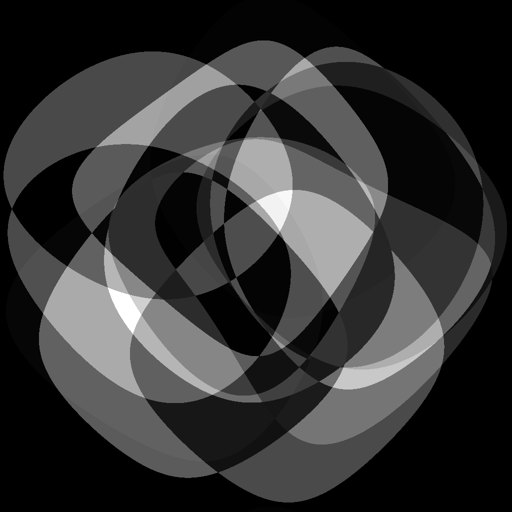} \\
$K_0 = 1$ & $K_0 = 5$ & $K_0 = 9$ & $K_0 = 13$ 
\end{tabular}}}
\end{tabular}
\end{minipage}
\caption{Phase transition experiments. We generated random piecewise constant images with known edge-set bandwidth and study the success rate proposed structured low-rank matrix completion scheme under two conditions: in (a) we vary the filter size $\Lambda_1$ while keeping the edge-set bandwidth $K_0$ fixed, in (b) we vary the edge-set bandwidth $K$ while keeping the filter size fixed. Examples of the randomly generated data are shown in (c).}
\label{fig:phasetransition}
\end{figure}
\subsection{Comparison with TV minimization on real MRI data}
We also compare the proposed Fourier domain interpolation scheme against standard discrete TV minimization in spatial domain: 
\begin{equation}
\label{tvrecon}
\min_{\bm u \in \mathbb{C}^{N}}~ TV(\bm u) ~\text{subject to}~ P_{\Omega}(\bm F \bm u) = P_{\Omega}(\bm F \bm u_0).
\end{equation}
Here $\bm u\in\mathbb{C}^{N}$ with $N=N_xN_y$ is a 2-D array representing a discrete $N_x\times N_y$ image, $\bm u_0\in\mathbb{C}^{N}$ is the image to be recovered, $\bm F \in \mathbb{C}^{N\times N}$ denotes the unitary 2-D discrete Fourier transform (DFT) matrix acting on $N_1\times N_2$ arrays, $\bm P_\Omega$ is projection onto the index of sampling locations $\Omega \subset [N_x] \times [N_y]$, and $TV(\cdot)$ denotes the (isotropic) total variation semi-norm:
\begin{equation}
TV(\bm u) = \sum_{i=1}^N(|(\bm \partial_1 \bm u)_{i}|^2 + |(\bm \partial_2 \bm u)_{i}|^2)^{\frac 1 2}
\end{equation}
where $\bm \partial_1$ and $\bm \partial_2$ are finite difference operators in the horizontal and vertical directions, respectively. The problem \eqref{tvrecon} has been studied extensively \cite{candes2006robust,candes2006stable,needell2013near,needell2013stable,krahmer2014stable,poon2015role} as a model for undersampled MRI reconstruction and other inverse problems in imaging. 


In Fig.~\ref{fig:mainresult} we perform an experiment comparing against TV minimization and the proposed approach on real MRI data. For this experiment we used a fully-sampled four-coil single-slice acquisition consisting of $256\times 256$ Cartesian k-space samples, which was compressed to a single virtual coil using an SVD-based technique \cite{zhang2013coil}. The data in the single virtual coil was observed to have smoothly varying complex phase in image domain. To compensate for this source of model-mismatch, we further pre-preprocessed the data by removing the complex phase in image domain. We note that this pre-processing step is unrealistic for a true MRI experiment. However, the optimization problem \eqref{nucnorm} could be modified to incorporate an estimate of the smoothly varying phase in the measurement model; we omit this step for simplicity. Finally, we retrospectively undersampled the pre-processed virtual single coil data, taking $50\%$ uniform random samples. We find that the proposed structured low-rank recovery shows significant improvement recovery error over standard total variation as measured by $\text{SNR} = 20\log_{10}(\|\hat f\|/\|\hat f^\ast-\hat f\|)$, where $\hat f^*$ is the recovered data and $\hat f$ is the ground truth. The error images indicate the proposed method more faithfully recovers the true edges of the image.

\begin{figure}[ht!]
\begin{minipage}{\linewidth}
\hspace{-1em}
\begin{tabular}{ccc}
Fully sampled & 
TV min. &
Proposed \\
\begin{minipage}{0.315\linewidth}
\begin{tikzpicture}
    \node[anchor=south west,inner sep=0] (image) at (0,0) {\includegraphics[height=\linewidth,width=0.9\linewidth,trim=60 70 150 70,clip,angle=-90,origin=c]{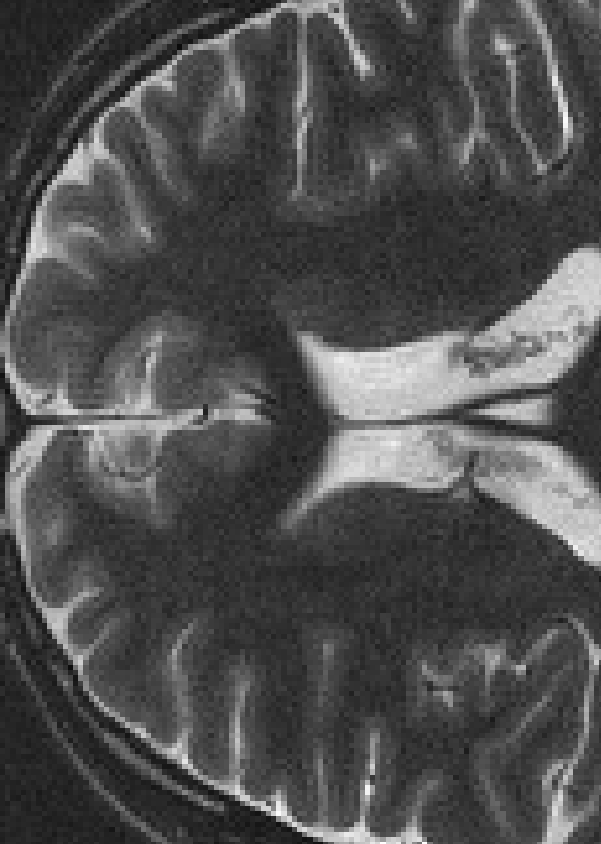}};
    \begin{scope}[x={(image.south east)},y={(image.north west)}]
        \draw[yellow,thick,-latex] (0.25,0.75) -- (0.40,0.80);
        \draw[yellow,thick,-latex] (0.42,0.2) -- (0.27,0.2);
    \end{scope}    
\end{tikzpicture}
\end{minipage}
& 
\hspace{-1em}
\begin{minipage}{0.315\linewidth}
\begin{tikzpicture}
    \node[anchor=south west,inner sep=0] (image) at (0,0) {\includegraphics[height=\linewidth,width=0.9\linewidth,trim=60 70 150 70,clip,angle=-90,origin=c]{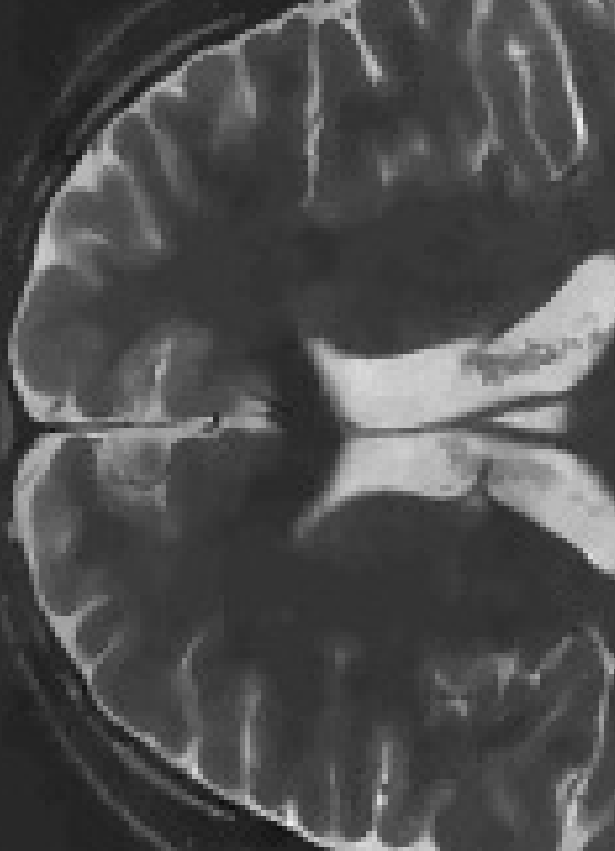}};
    \begin{scope}[x={(image.south east)},y={(image.north west)}]
        \draw[yellow,thick,-latex] (0.25,0.75) -- (0.40,0.80);
        \draw[yellow,thick,-latex] (0.42,0.2) -- (0.27,0.2);
    \end{scope}         
\end{tikzpicture}
\end{minipage}
& 
\hspace{-1em}
\begin{minipage}{0.315\linewidth}
\begin{tikzpicture}
    \node[anchor=south west,inner sep=0] (image) at (0,0) {\includegraphics[height=\linewidth,width=0.9\linewidth,trim=60 70 150 70,clip,angle=-90,origin=c]{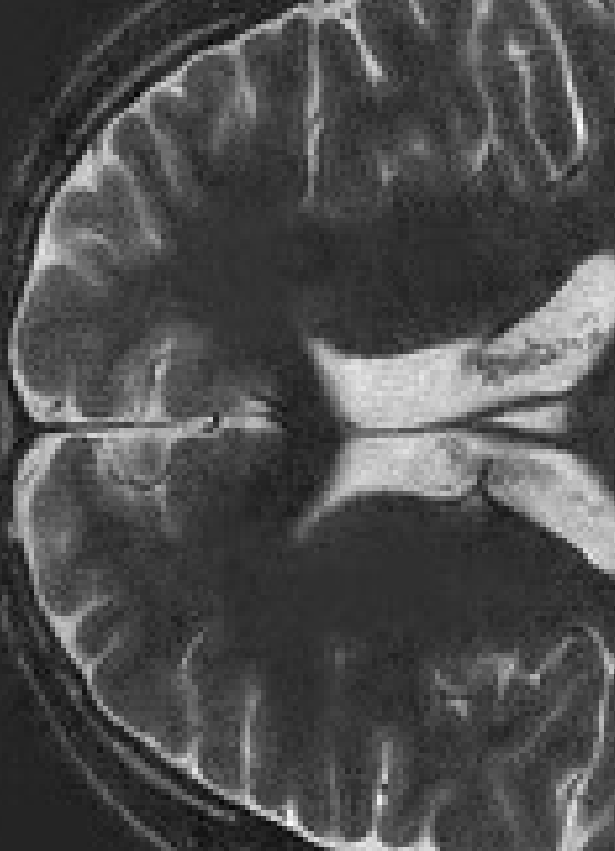}};
    \begin{scope}[x={(image.south east)},y={(image.north west)}]
        \draw[yellow,thick,-latex] (0.25,0.75) -- (0.40,0.80);
        \draw[yellow,thick,-latex] (0.42,0.2) -- (0.27,0.2);
    \end{scope}      
\end{tikzpicture}
\end{minipage}
\\
\begin{minipage}{0.315\linewidth}
\centering
\begin{tikzpicture}
    \node[anchor=south west,inner sep=0] (image) at (0,0) {\includegraphics[width=0.9\linewidth,height=0.75\linewidth,angle=-90,origin=c]{images/realbrain_cor2_hi_uniform0p5_redo_orig_gray.pdf}};
    \begin{scope}[x={(image.south east)},y={(image.north west)}]
         \draw[yellow,thick] (0.13,0.41) rectangle (0.85,0.85);
    \end{scope}  
\end{tikzpicture}
\end{minipage}
& 
\hspace{-1em}
\begin{minipage}{0.315\linewidth}
\begin{tikzpicture}
    \node[anchor=south west,inner sep=0] (image) at (0,0) {\includegraphics[height=\linewidth,width=0.9\linewidth,trim=60 70 150 70,clip,angle=-90,origin=c]{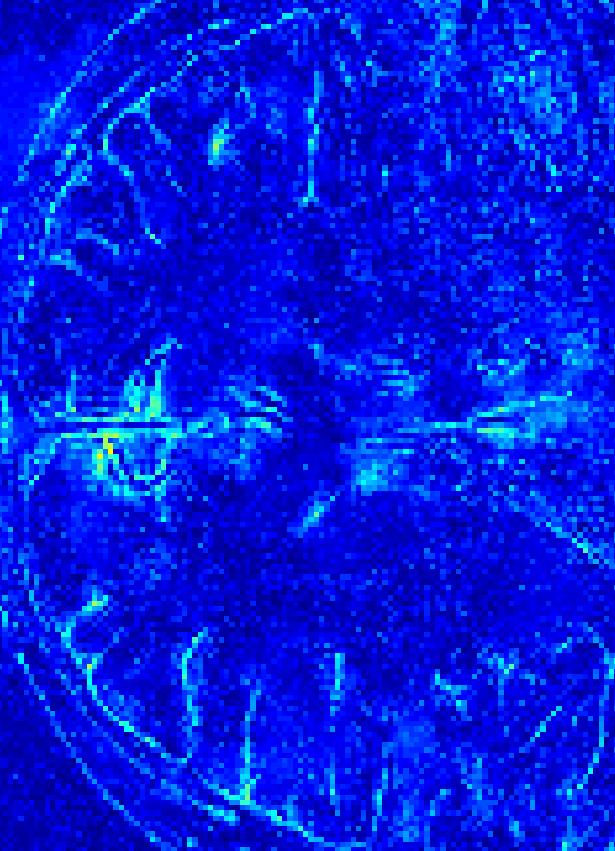}};
    \begin{scope}[x={(image.south east)},y={(image.north west)}]
        \draw[yellow,thick,-latex] (0.25,0.75) -- (0.40,0.80);
        \draw[yellow,thick,-latex] (0.42,0.2) -- (0.27,0.2);
    \end{scope}    
\end{tikzpicture}
\end{minipage}
&
\hspace{-1em}
\begin{minipage}{0.315\linewidth}
\begin{tikzpicture}
    \node[anchor=south west,inner sep=0] (image) at (0,0) {\includegraphics[height=\linewidth,width=0.9\linewidth,trim=60 70 150 70,clip,angle=-90,origin=c]{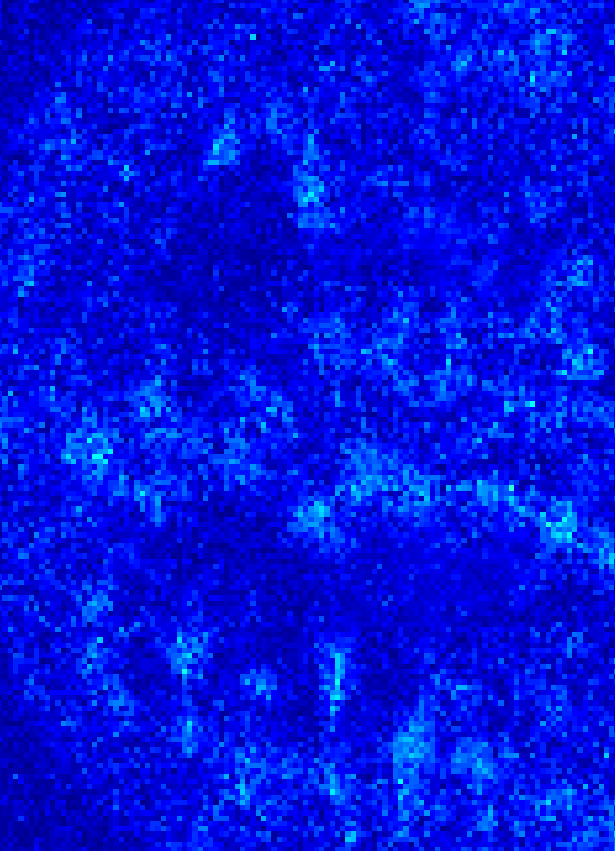}};
    \begin{scope}[x={(image.south east)},y={(image.north west)}]
        \draw[yellow,thick,-latex] (0.25,0.75) -- (0.40,0.80);
        \draw[yellow,thick,-latex] (0.42,0.2) -- (0.27,0.2);
    \end{scope}
\end{tikzpicture}
\end{minipage}
\\
 & 
SNR=17.8dB &  
SNR=19.0dB
\end{tabular}
\end{minipage}
\caption{Recovery of MRI data from 2-fold random uniform undersampling. Error images shown in bottom right.}
\label{fig:mainresult}
\end{figure}

\section{Discussion}
Discrete domain total-variation minimization has played a central role in compressed sensing from its inception \cite{candes2006robust,candes2006stable}, which models the image to be recovered as (approximately) piecewise constant. Since the present work can be thought of as an extension of compressed sensing type guarantees to the continuous domain setting, it is fruitful to explore the connections between our continuous domain model and discrete domain total variation. 

At first glance, the structured low-rank matrix completion problem \eqref{nucnorm} may seem far removed from the TV-minimization problem \eqref{tvrecon}. But, in fact, one can show TV-minimization \eqref{tvrecon} is equivalent to nuclear norm minimization of a related structured matrix lifting in Fourier domain. Specifically, \eqref{tvrecon} is equivalent to
\begin{equation}\label{eq:tvlifting}
	\min_{\bm v}\|\mathcal{C}(\bm F\bm u )\|_* ~\text{subject to}~ P_\Omega(\bm F \bm u) = P_\Omega(\bm F \bm u_0).
\end{equation}
Here \begin{equation}
\mathcal{C}(\bm F\bm u ) = 
\begin{bmatrix}
	\mathcal{C}_x(\bm F\bm u )\\
	\mathcal{C}_y(\bm F\bm u )
\end{bmatrix}\in\mathbb{C}^{2N\times N}
\end{equation}
and $\mathcal{C}_x(\bm F\bm u)$, $\mathcal{C}_y(\bm F\bm u)$ are block circulant with circulant blocks matrices whose first column is specified by the arrays $\bm v_x = \bm F\bm \partial_x \bm u$ and $\bm v_y = \bm F\bm \partial_x \bm u$. Assuming circular boundary conditions, we can write $(\bm v_x)[k_x,k_y] = (1-e^{j2\pi k_x/N_x})(\bm F\bm u)[k_x,k_y]$ and $(\bm v_y)[k_x,k_y] = (1-e^{j2\pi k_y/N_y})(\bm F\bm u)[k_x,k_y]$.  

We find it interesting to use this re-formulation of TV-minimization to better understand the proposed approach. In Table \ref{table:compare} we summarize the similarities and differences. One essential difference is the dimensions of the matrix liftings. In particular, the matrix lifting we propose has dimensions $2|\Lambda_2|\times|\Lambda_1|$, with $|\Lambda_1| \ll |\Lambda_2|$ whereas the matrix lifting associated with TV in \eqref{eq:tvlifting} has dimensions $2N\times N$. If the reconstruction grid size is the same in both cases, i.e., $|\Gamma| = N$, then the proposed matrix lifting has substantially fewer columns than the one associated with TV. This is due to our assumption that edge set of the image has low bandwidth. In other words, we restrict the degrees of freedom of the model by constructing a lifting with fewer columns. We believe this difference may explain the success of the proposed method over TV-minimization observed empirically in Section \ref{sec:exp}.

\begin{table}
\centering
\begin{tabular}{c|c|c}
 & TV-minimization & Proposed\\
 \hline
Spatial domain & discrete & continuous\\
Derivative operator & finite differences & exact derivative\\
Singularity set & discrete points & connected curves\\
\hline
Frequency domain & discrete & discrete\\
Frequency weighting $w_i[\bm k]$ & $1-e^{j2\pi k_i/N_i}$ & $j2\pi k_i$\\
Lifted matrix structure & two-level circulant & two-level Toeplitz\\
Rank of lifted matrix & sparsity of & bandwidth of \\
   & discrete gradient & edge set\\
\hline
\end{tabular}
\vspace{1em}
\caption{Comparison of proposed scheme with discrete total variation minimization}
\label{table:compare}
\vspace{-1em}
\end{table}
\section{Conclusion}
We derived performance guarantees for the recovery of piecewise constant images from random non-uniform Fourier samples via a convex  structured low-rank matrix completion problem. This was achieved by adapting results in \cite{chen2014robust} to the case of a low-rank block two-fold Toeplitz matrix with an additional weighting scheme that arises naturally when considering piecewise constant images. We also define incoherence measures that rely only on geometric properties of the edge set, which indicate that the sampling burden is higher for images with smaller piecewise constant regions.


The recovery guarantees in this work studied the case of uniform random samples. However, in practice we observe that recovery works well with when considering other types of variable density random sampling, where the low spatial frequencies are more heavily sampled. It would be interesting to adapt our results to a wider variety of sampling distributions, and to identify the optimal sampling strategy for signals belonging to our image model.

\section{Appendix A: Incoherence Bounds}\label{sec:appendixA}
\setcounter{subsection}{0}
\subsection{Notation and Preliminaries}

To simplify our arguments, we will convert the linear operators $\Tf$ and $\Tf^*$ defined in Fourier domain to linear operators acting on spaces of trigonometric polynomials \eqref{eq:trigpoly} in spatial domain. Specifically, for any index set $\Omega \subset \mathbb{Z}^2$, let $B_\Omega$ denote the vector space of all trigonometric polynomials that have coefficients supported within $\Omega$. Similarly, we denote the space of vector fields $\bm \rho = (\rho_1,\rho_2)$ with components $\rho_1,\rho_2 \in B_\Omega$ as $B^2_\Omega$. We set $\mathcal{S}(f)= \mathcal{F}\Tf\mathcal{F}^{-1}$, where $\mathcal{F}$ is the Fourier transform of a periodic function on $[0,1]^2$. For any index set $\Lambda$, define the Dirichlet kernel $D_{\Lambda_1}(\bm r) := \sum_{\bm k \in \Lambda_1} e^{j2\pi \bm k \cdot \bm r}$. For all $\varphi \in B_{\Lambda_1}$, the action of the linear operator $\mathcal{S}(f):B_{\Lambda_1}\rightarrow B_{\Lambda_2}^2$ can be expressed compactly as
\begin{equation}
\mathcal{S}(f)\varphi = D_{\Lambda_2}\ast(\varphi\,\nabla f)\in B_{\Lambda_1}^2,
\end{equation}
where $\varphi\,\nabla f$ is understood as a tempered distribution, and the convolution is applied separately to each vector field component. Here convolution with $D_{\Lambda_2}$ is a bandlimiting operation. Simliarly, for $\bs \rho =(\rho,\rho_2) \in B^2_{\Lambda_2}$, the adjoint $\mathcal{S}(f)^*$ acts as
\begin{equation}
\mathcal{S}(f)^*\bs\rho = D_{\Lambda_1}\ast(\bs\rho \cdot {\nabla f}) \in B_{\Lambda_1}
\end{equation}
which is the spatial domain equivalent of the adjoint matrix $\Tf^*$. More expliclty, if $f = 1_U$ where $U$ is a simply connected region with smooth boundary $\partial U$, a straightforward argument using the divergence theorem shows that the function $\mathcal{S}(f)\varphi$ is given pointwise as the weighted curve integral
\begin{equation}
(\mathcal{S}(f)\varphi)(\bm r) = \oint_{\partial U} D_{\Lambda_2}(\bm r-\bm r')\,\bm n(\bm r')\,ds(r'),
\end{equation}
for all $\bm r \in [0,1]^2$, where $\bm n(\bm r')$ is the outward unit normal to the curve $\partial U$ at $r'$, and $ds$ is the arc-length element. Likewise, $\mathcal{S}(f)^*\bm \rho$ is the function given pointwise by
\begin{equation}\label{eq:Sfstar}
(\mathcal{S}(f)^*\bm \rho)(\bm r) = \oint_{\partial U} D_{\Lambda_1}(\bm r-\bm r')\,[\bm\rho(\bm r')\cdot\bm n(\bm r')]\,ds(r'),
\end{equation}
for all $\bm r \in [0,1]^2$. These formulas can be generalized to an arbitrary piecewise constant function $f = \sum_{i} a_i 1_{U_i}$ by linearity. However, in the remainder we focus on the case where $f = 1_U$ to simplify our arguments.

\subsection{Fundamental subspaces of $\mathcal{S}(f)$ and dimensions}
Under the conditions of Theorem \ref{prop:rank}, the nullspace of $\Tf$ is spanned by shifts of the minimal annihilating filter, $\widehat{\mu_0}$. In spatial domain, this space consists of all multiples of the minimal degree polynomial $\gamma = \eta\,\mu_0$ such that $\gamma$ is bandlimited to $\Lambda_1$. We denote this space by
\begin{equation}
(\mu_0)_{\Lambda_1} := \{\eta\,\mu_0 : \eta \in B_{\Lambda_1\colon\Lambda_0}\}.
\end{equation}
Note that $(\mu_0)_{\Lambda_1}$ is a subspace of $B_{\Lambda_1}$ with dimension $|\Lambda_1\colon\Lambda_0|$. Therefore, the dimension of the kernel of $\mathcal{S}(f)$, denoted by $\ker \mathcal{S}(f)$, is given by
\begin{equation}
\dim \ker \mathcal{S}(f) = |\Lambda_1\colon\Lambda_0|.
\end{equation} 
By the rank-nullity theorem, the dimension of the image of $\mathcal{S}(f)$, denoted by $\text{im}\,\mathcal{S}(f)$, is
\begin{equation}
\dim \text{im}\,\mathcal{S}(f) =  |\Lambda_1| - |\Lambda_1\colon\Lambda_0| = R.
\end{equation}
Likewise, the dimension of the coimage $\text{im}\,\mathcal{S}(f)^*$ is also $R$. Furthermore, since $\text{im}\,\mathcal{S}(f)^* = [\ker S(f)]^\perp$, we have
\begin{equation}
\text{im}\,\mathcal{S}(f)^* = (\mu_0)^\perp_{\Lambda_1}
\end{equation}
This means that any $\gamma \in B_{\Lambda_1}$ is in the row space if and only if $\gamma$ is orthogonal to every trigonometric polynomial of the form $\eta\,\mu_0 \in  B_{\Lambda_1}$, or equivalently,
\begin{equation}
\langle \gamma, \eta\,\mu_0 \rangle = \int_{[0,1]^2}\gamma(\bm r) \overline{\eta(\bm r) \mu_0(\bm r)}\,d\bm r = 0
\end{equation}
for all $\eta \in B_{\Lambda_1\colon\Lambda_0}$.


\subsection{Basis for the coimage of $\mathcal{S}(f)$ (corresponding to the row space of $\mathcal T(\hat f)$)}
\label{rowspaceappendix}
Let $\bm s \in [0,1]^2$, and set $\varphi_{\bm s} \in B_{\Lambda_1}$ to be the translated Dirichlet kernel: 
\begin{equation}
\varphi_{\bm s}(\bm r) = D_{\Lambda_1}(\bm r - \bm s)~~\text{for all}~~\bm r \in [0,1]^2.
\end{equation}
Equivalently, $\varphi_{\bm s} \in B_{\Lambda_1}$ is the trigonometric polynomial specified in Fourier domain as
\begin{equation}
\widehat{\varphi_{\bm s}} [\bm k] = 
\begin{cases}
e^{-j2\pi \bm s\cdot \bm k} & \text{if }\bm k \in \Lambda_1 \\
0 & \text{if }\bm k \not\in \Lambda_1
\end{cases}.
\end{equation}
Observe that the inner product of $\varphi_{\bm s}$ with any other trigonometric polynomial $\eta \in B_{\Lambda_1}$ is given by the point-evaluation of $\eta$ at $\bm s$:
\begin{equation}
\langle \eta, \varphi_{\bm s} \rangle = \sum_{\bm k \in \Lambda_1} \hat\eta[\bm k] e^{j2\pi \bm k \cdot \bm s} = \eta(\bm s). 
\end{equation}
Suppose now that the point $\mbf s$ satisfies $\mu_0(\bm s) = 0$. In this case, we see that $\varphi_{\bm s}$ is necessarily in the coimage $\text{im}\,\mathcal{S}(f)^* = (\mu_0)_{\Lambda_1}^\perp$ since we have
\begin{align}
\label{orthog}
\langle \gamma\mu_0, \varphi_{\bm s} \rangle = \gamma(\mbf s)\mu_0(\mbf s) = 0.
\end{align}
for any multiple of the minimal polynomial $\gamma \mu_0 \in B_{\Lambda_1}$, i.e., any element in $\ker \mathcal{S}(f) = (\mu_0)_{\Lambda_1}$. 

We will now show how to construct a basis for the coimage of $\mathcal{S}(f)$ out of elements having the form $\varphi_{\bm r_i}$ for some $\bm r_i$, $i=1,...,R$ belonging to the zero set of $\mu_0$. For an arbtirary collection of $R$ points $\{\bm r_i\}_{i=1}^R\subset \{\mu_0 = 0\}$, we are not guaranteed that the set of functions $\{\varphi_{\bm r_i}\}_{i=1}^R$ is linearly independent. However, we will show that there exists a constant $M = M(\Lambda_0,\Lambda_1)$ such that \emph{for any} $M$ distinct points $\{\bm r_i\}_{i=1}^M \subset \{\mu_0 = 0\}$ we can always find a subset of $R$ linearly independent basis functions from the collection $\{\varphi_{\bm r_i}\}_{i=1}^M$. The constant $M$ is related the maximum number of isolated zeros that a system of two trigonometric polynomials can have. The following lemma, which is a consequence of the BKK bound in enumerative algebraic geometry (see, e.g., \cite{li1996bkk}), puts a bound on $M$. See section \ref{SMbkk} of the supplementary material for proof.


\begin{lem}\label{lem:bkk}
Let $\Lambda_1$ and $\Lambda_0$ be rectangular index sets such that $\Lambda_0\subset\Lambda_1$, and set $R=|\Lambda_1|-|\Lambda_1\colon\Lambda_0|$. For any $\mu_0,\mu_1$ trigonometric polynomials bandlimited to $\Lambda_0$ and $\Lambda_1$, respectively, the maximum number $M$ of isolated solutions of $\mu_0(\bm r) = \mu_1(\bm r) = 0$ is bounded as
\begin{equation}
\label{bound}
M < R+|\Lambda_0|.
\end{equation}
\end{lem}
We now prove equivalents of Lemma \ref{rowlemma} and Lemma \ref{lem:admissible} in terms of the spatial domain operator $\mathcal S(f)$:
\begin{lem} Let $\{\mbf r_1,....,\mbf r_N\}$ be any collection of $N$ distinct points on the curve $\{\mu_0=0\}$, where $N \geq R+|\Lambda_0|$. Then the coimage space $\text{im} \,\mathcal{S}(f)^* = (\mu_0)^\perp_{\Lambda_1}$ is spanned by the set of shifted Dirichlet kernels $\varphi_i(\bm r) = D_{\Lambda_1}(\bm r- \bm r_i)$ for all $i=1,...,N$, i.e.,
\begin{equation}
\label{equality}
\text{span}\brac{\varphi_{\bm r_i}}_{i=1}^N = (\mu_0)_{\Lambda_1}^\perp.
\end{equation}
In particular, there exists a subset of $R = |\Lambda_1|-|\Lambda_1\colon\Lambda_0|$ elements of $\brac{\varphi_{\bm r_i}}_{i=1}^N$ that is a basis for the coimage space $(\mu_0)_{\Lambda_1}^\perp$.
\end{lem}
\begin{proof}
All the functions $\varphi_{\bm r_i}$ are in $(\mu_0)_{\Lambda_1}^{\perp}$ since we have $\inner{\varphi_i,\gamma\mu_0}= \gamma(\mbf r_i)\mu_0(\mbf r_i)=0$ because each $\mbf r_i$ belong to the zero set of $\mu_0$. This implies that 
\begin{equation}
\label{subset}
\text{span}\brac{\varphi_{\bm r_i}}_{i=1}^M \subseteq (\mu_0)_{\Lambda_1}^\perp.
\end{equation}
Our focus is on proving \eqref{equality} with equality. For this, it is sufficient to show that any vector orthogonal to $\text{span}\{\varphi_i\}_{i=1}^N$ is in $(\mu_0)_{\Lambda}$. Assume that there is a vector $\eta(\mbf r) \in B_{\Lambda_1}$ that is in the orthogonal complement space of $\text{span}\{\varphi_i\}_{i=1}^N$. This is only possible if
\begin{equation}
\langle \eta, \varphi_i \rangle = \eta(\bm r_i) = 0,~~\text{for all}~~i=1,...,N.
\end{equation}
Therefore, both $\eta$ and $\mu_0$ have $N$ zeros in common. By Lemma \ref{lem:bkk} this is only possible if $\eta$ contains $\mu_0$ as a factor. This implies that all vectors in the orthogonal complement space of $\text{span}\{\varphi_i\}_{i=1}^M$ are in $(\mu_0)_{\Lambda_1}$, or equivalently 
\begin{equation}
\label{subset2}
\text{span}\brac{\varphi_{\bm r_i}}_{i=1}^M \supseteq (\mu_0)_{\Lambda_1}^\perp,
\end{equation}
which together with \eqref{subset} proves \eqref{equality}.

Finally, we also know that the dimension of $(\mu_0)_{\Lambda_1}^{\perp}$ is equal to $R<M$. Thus, one can select a subset of $R$ basis functions $\varphi_i$ that are linearly independent and hence a basis for $(\mu_0)^\perp_{\Lambda_1}$. 

\end{proof}


Translating this result to Fourier domain, we see that the row space of $\Tf$ is spanned by the vectors of Fourier coefficients $(\widehat{\varphi_i}[\bm k] : \bm k \in \Lambda_1) \in \mathbb{C}^{|\Lambda_1|}$, for $i=1,...,R$. Equivalently, this can be expressed as the columns of the Vandermonde-like matrix $\ER$ specified by \eqref{rowspace}, which proves Lemma \ref{rowlemma} and \ref{lem:admissible}.

\subsection{Discretization of curve integrals: quadrature formula}
Using the results from the previous subsection, we now introduce a quadrature formula for curve integrals, which we will use to determine the range space $\text{im}\,\mathcal{S}(f)$ in the next subsection. 

Let $\gamma$ be any function in  $B_\Lambda$ for any $\Lambda \supseteq \Lambda_0$. Then from the orthogonal decomposition $B_\Lambda = (\mu_0)_{\Lambda}\oplus(\mu_0)^\perp_{\Lambda}$ we can decompose $\gamma$ as
\begin{equation}\label{eq:orthdecomp}
\gamma(\bm r) = \sum_{i=1}^S a_i D_{\Lambda}(\bm r - \bm r_i) + \varphi(\bm r)\mu_0(\bm r),
\end{equation}
where $S = |\Lambda|-|\Lambda\colon\Lambda_0|$, and where $\{D_\Lambda(\bm r - \bm r_i)\}_{i=1}^S$ defines a basis of $(\mu_0)_{\Lambda}^\perp$. Here, the coefficients $a_i$ in \eqref{eq:orthdecomp} are obtained uniquely as
\begin{equation}
\begin{bmatrix}
a_1\\ \vdots \\ a_S 
\end{bmatrix}
= \bm D^{-1}
\begin{bmatrix}
\gamma(\bm r_1)\\ \vdots \\ \gamma(\bm r_S) 
\end{bmatrix},
\end{equation}
where $\bm D \in \mathbb{R}^{S\times S}$ is the symmetric matrix with entries $[\bm D]_{i,j} = D_\Lambda(\bm r_i - \bm r_j)$ for $1 \leq i,j\leq S$.
The above expression can be compactly expressed as $
\bm a = \bm D^{-1} \bm g$, where $\bm g = (\gamma(\bm r_1),...,\gamma(\bm r_S))^T$. 

\begin{lem}
\label{agrlemma}
 Let $f = 1_U$ where $U$ is a simply connected region with smooth boundary $\partial U$, which is the zero levelset of $\mu_0 \in B_{\Lambda_0} $ and let $\gamma \in B_{\Lambda}$. Consider the curve integral of the form
\begin{equation}
\mbf q = \oint_{\partial U} \gamma(\bm r)\,\bm n(\bm r)\,ds(\bm r),
\end{equation}
where $\bm n(\bm r) = \nabla f(\mbf r)/|\nabla f(\mbf r)|$ is the unit normal on the curve $\partial U$. The curve integral can be evaluated using the quadrature formula
\begin{equation}
\label{quad}
\mbf q  =\sum_{i=1}^S \gamma(\bm r_i) ~\bm w_i,
\end{equation}
where the $S= |\Lambda|-|\Lambda:\Lambda_0|$ points $\{\bm r_i\}_{i=1}^S$ belong to the curve $\{\mu_0 = 0\}$, and the cooresponding weight vectors $\mbf w_i \in \mathbb{R}^{2}$, $i=1,..,S$, are specified by 
\begin{equation}
\label{aggregated}
\begin{bmatrix}
\bm w_1\\
\vdots\\
\bm w_S
\end{bmatrix}
 = \bm D^{-1} \begin{bmatrix} \bm v_1 \\ \vdots \\ \bm v_S \end{bmatrix}.
\end{equation}
where $\bm v_i = \oint_{\partial U} D_\Lambda(\bm r - \bm r_i) \bm n(\bm r) ds(\bm r) \in \mathbb{R}^2$.
 \end{lem}
 \begin{proof}
 Decomposing $\gamma(r)$ using \eqref{eq:orthdecomp}, we obtain
\begin{equation}
\oint_{\partial U} \gamma(\bm r)\,\bm n(\bm r)\,ds(\bm r) = \sum_{i=1}^S a_i \underbrace{\oint_{\partial U} D_\Lambda(\bm r-\bm r_i)\,\bm n(\bm r)\,ds(\bm r)}_{:=\bm v_i}\\
\end{equation}
The above sum can be expressed in the vector form as 
\begin{equation}
\label{quadrature}
\sum_{i=1}^S a_i \bm v_i = \bm a^* \bm V = \bm g^* \bm D^{-1} \bm V
\end{equation}
where $\bm V = [\bm v_1^T,...,v_S^T]^T \in \mathbb{C}^{R\times 2}$. Setting $\bm W = \bm D^{-1} \bm V = [\bm w_1^T,...,\bm w_S^T]^T \in \mathbb{C}^{R\times 2}$ we obtain \eqref{quad}.
\end{proof}



\subsection{Basis for the range of $\mathcal{S}(f)$ (corresponding to the column space of $\Tf$)}
\label{colspaceappendix}
We now introduce a basis set for $\mathsf{im}\,\mathcal{S}(f)$, which will be used to prove Lemma \ref{columnlemma}.
\begin{lem}
The range of  $\mathcal S(f)$, denoted by $\mathsf{im}\,\mathcal{S}(f)$ is specified by
\begin{equation}
\mathsf{im}\,\mathcal{S}(f) = \text{span}\{ \bm w_i \, D_{\Lambda_2}(\bm r-\bm r_i)\}_{i=1}^R
\end{equation}
for an appropriate choice of points $\{\bm r_i\}_{i=1}^{R} \subset \{\mu_0 = 0\}$ with $R = |\Lambda_1|-|\Lambda_1\colon\Lambda_0|$, and where the weight vectors $\bm w_i$ are specified by \eqref{aggregated}.
\end{lem}

\begin{proof}
Consider an arbitrary element $\bs \rho =(\rho,\rho_2) \in \mathsf{im}\,\mathcal{S}(f)$. We can express $\bs\rho$ as $\bs\rho = \mathcal{S}(f)\psi = \mathcal{B}_{\Lambda_2}(\psi \nabla f) = D_{\Lambda_2} \ast (\psi\nabla f)$ for some $\psi\in B_{\Lambda_1}$. By the definition in \eqref{eq:Sfstar}, we have
\begin{align}\nonumber
 \bs\rho(\bm r) 
& = \oint_{\partial U} {\psi(\bm s)}\,D_{\Lambda_2}(\bm r-\bm s)\,\bm n(\bm s)\,ds\\\label{lemma4.1}
& = \sum_{i=1}^{S} {\psi(\bm r_i)} \, D_{\Lambda_2}(\bm r-\bm r_i)\,\bm w_i,
\end{align}
where we 
Lemma \ref{agrlemma} in the last step with $S = |\Gamma|-|\Gamma:\Lambda_0|$ since the integrand ${\psi(\bm s)}^*\,D_{\Lambda_2}(\bm r-\bm s)$ belongs to $B_\Gamma$. The above relation shows that any $\bm \rho(\mbf s) \in \mathsf{im}\,\mathcal{S}(f)$ can be expressed as the linear combination of the functions $D_{\Lambda_2}(\bm s-\bm r_i)\,\bm w_i,$ for $i=1,..,S$. Thus, we have $\mathsf{im}\,\mathcal{S}(f) \subset \text{span}\{D_{\Lambda_2}(\bm r-\bm r_i)\,\bm w_i\}_{i=1}^{S}$. We also know that $\dim\bkt{\mathsf{im}\,\mathcal{S}(f)} = R <S$. This implies that we can select a subset of $R$ vectors from the set $\{D_{\Lambda_2}(\bm r-\bm r_i)\,\bm w_i\}_{i=1}^{S}$ that are linearly independent, which will span $\mathsf{im}\,\mathcal{S}(f)$, and hence define a basis.
\end{proof}

Correspondingly, the column space of $\mathcal T(\hat f)$ is spanned by the Fourier coefficients of the basis vectors $\mbf w_i D_{\Lambda_2}(\mbf r-\mbf r_i)$, or the columns of the $2|\Lambda_2| \times R$ weighted Vandermonde-like matrix $\EL$ specified by \eqref{columnspace}. 

\subsection{Incoherence Bounds}
\label{coherenceappendix}
\subsubsection{Projection onto row subspace}
Let $\ER = \ER(P)$ be any basis for the row space $V$ of $\Tf$ specified by \eqref{rowspace}, whose columns are vectorized Fourier coefficients of the translated and normalized Dirichlet kernels $\varphi_i(\bm r) = \frac{1}{\sqrt{|\Lambda_1|}}D_{\Lambda_1}(\bm r-\bm r_i)$, $i=1,...,R$, for some set of admissible nodes $P=\{\bm r_1,...,\bm r_R\}\subset \{\mu_0 = 0\}$. 
Projecting the measurement basis matrix $\bm A_{\bm k}$ onto $V$, we have
\begin{align*}
\|\mathcal{P}_V \mbf A_{\mbf k}\|_F^2 & = \|\mbf A_{\mbf k} \ER (\ER^* \ER)^{-1} \ER^* \|_F^2\\
& \leq [\lambda_{min} (\ER^*\ER)]^{-1} \|\mbf A_{\mbf k} \ER\|_F^2 
\end{align*}
Since $\mbf A_{\mbf k}$ selects $|\omega(\bm k)|$ rows of $\ER$, each of which has $R$ entries of magnitude $1/\sqrt{|\Lambda_1|}$, we have
\begin{equation}
\|\mbf A_{\mbf k} \ER\|_F^2 = \frac{1}{|\omega(\mbf k)|}\cdot R\cdot|\omega(\mbf k)|\cdot \frac{1}{|\Lambda_1|} = \frac{R}{|\Lambda_1|} = \frac{R\,c_s}{|\Gamma|}
\end{equation}
where $c_s = |\Gamma|/|\Lambda_1|$. Hence, 
\begin{equation}
\|\mathcal{P}_V \mbf A_{\mbf k}\|_F^2 \leq [\lambda_{min} (\ER^*\ER)]^{-1} \frac{R~c_s }{|\Gamma|}.	
\end{equation}
Minimizing over all sets of admissible nodes $P$ in the construction of $\ER$ gives the final bound
\begin{equation}
\|\mathcal{P}_V \mbf A_{\mbf k}\|_F^2 \leq \frac{\rho\, R\,c_s }{|\Gamma|}.
\end{equation}
\subsubsection{Projection onto column space}
Let $\EL = \EL(P)$ be a basis for the column space of $\Tf$ specified by \eqref{columnspace}, whose columns are vectorized Fourier coefficients of the translated and weighted Dirichlet kernels $\frac{1}{\sqrt{|\Lambda_2|}}\frac{\bm w_i}{\|\bm w_i\|} D_{\Lambda_2}(\bm r-\bm r_i)$, for some set of admissible nodes $P=\{\bm r_1,...,\bm r_R\}\subset \{\mu_0 = 0\}$. Observe the columns of $\EL$ are defined to have unit $\ell^2$-norm. Following the same steps as in the row space bound, we have
\begin{align*}
\|\mathcal{P}_U \bm A_{\bm k}\|_F^2 & = \|\EL (\EL^* \EL)^{-1} \EL^* \bm A_{\bm k} \|_F^2 \\
& \leq [\lambda_{min} (\EL^*\EL)]^{-1} \|\EL^* \bm A_{\bm k} \|_F^2
\end{align*}
Expanding the norm $\|\EL^*\bm A_{\bm k}\|_F^2$ gives
\begin{align*}
\|\EL^*\bm A_{\bm k}\|_F^2 & = \frac{1}{|\Lambda_2|}\sum_{i=1}^{R}\frac{1}{|\omega(\bm k)|}\sum_{\bm \ell \in \omega(\bm k)}\left|\left\langle \frac{\bm \ell}{\|\bm \ell\|}, \frac{\bm w_i}{\|\bm w_i\|} \right\rangle\right|^2 \\
& \leq \frac{R}{|\Lambda_2|} \leq \frac{R\,c_s}{|\Gamma|}.
\end{align*}
Hence, we have
\begin{equation}\label{eq:rhoprimebnd}
\|\mathcal{P}_U \bm A_{\bm k}\|_F^2 \leq \frac{\rho'\,R\,c_s}{|\Gamma|}.
\end{equation}
where $\rho'$ is defined similarly to $\rho$ as:
\begin{equation}\label{eq:incoherence_left}
\rho' = \min_{\substack{P \subset \{\mu_0=0\}\\|P|=R}} \frac{1}{\lambda_{min}[\EL(P)^*\EL(P)]},
\end{equation}

Finally, we show how to bound $\rho'$ by $\rho$ in \eqref{eq:rhoprimebnd}. Observe that we can re-define $\rho$ and $\rho'$ in terms of the minimum singular value of the basis matrices $\ER(P)$ and $\EL(P)$, according to the correspondences:
\begin{align*}
\lambda_{min}(\EL(P)^*\EL(P)) & = \sigma_{min}^2(\EL(P)),\\
\lambda_{min}(\ER(P)^*\ER(P)) & = \sigma_{min}^2(\ER(P)).
\end{align*}
We will show $\sigma^2_{min}(\ER(P)) \leq \sigma^2_{min}(\EL(P))$, or equivalently, $[\lambda_{min}(\EL(P)^*\EL(P))]^{-1} \leq [\lambda_{min}(\bm \ER(P)^*\ER(P)]^{-1}$, for any set $P$ consisting of $R$ points on the edge set. The claim then follows immediately by taking the minimum over all such sets $P$. 

To ease notation, we drop the dependence on the set $P$ in the following. Observe that we can express $\EL$ as 
\begin{equation}
\EL = 
\begin{bmatrix}
\tilde{\bm E}_{\text{col}} \bm W_x\\
\tilde{\bm E}_{\text{col}} \bm W_y
\end{bmatrix}
\end{equation}
where $\bm W_x = \text{diag}(\frac{w_{1,x}}{\|\bm w_1\|},...,\frac{w_{R,x}}{\|\bm w_R\|})$, $\bm W_y = \text{diag}(\frac{w_{1,y}}{\|\bm w_1\|},...,\frac{w_{R,y}}{\|\bm w_R\|})$, and  $\tilde{\bm E}_{\text{col}} \in \mathbb{C}^{|\Lambda_2|\times R}$ is the Vandermonde-like matrix given entrywise by $[\tilde{\bm E}_{\text{col}}]_{i,j} = e^{j2\pi \bm k_{i}\cdot\bm r_i},~\text{for all}~~\bm k_i\in\Lambda_2, 1\leq j\leq R$. In other words, $\tilde{\bm E}_{\text{col}}$ has the same structure as $\ER$, but is built with respect to $\Lambda_2$ instead of $\Lambda_1$. In particular, since we always assume $\Lambda_1\subset\Lambda_2$, the matrix $\ER$ can be embedded as a submatrix of $\tilde{\bm E}_{\text{col}}$ by restricting the rows of $\tilde{\bm E}_{\text{col}}$ to those indexed by $\Lambda_1$. By the variational characterization of the minimum singular value of a matrix, we have
\begin{align}
\nonumber \sigma_{min}^2(\EL) & = \min_{\|\bm u\|=1} \|\EL \bm u\|^2\\ 
\nonumber & = \min_{\|\bm u\|=1} \|\tilde{\bm E}_{\text{col}} \bm W_x\bm u\|^2 + \|\tilde{\bm E}_{\text{col}} \bm W_y \bm u\|^2\\
& \geq \sigma_{min}^2(\tilde{\bm E}_{\text{col}})\underbrace{(\|\bm W_x \bm u\|^2 + \|\bm W_y\bm u\|^2)}_{=1} 
\label{eq:singvalineq}
\end{align}
Finally, since $\ER$ is a submatrix of $\tilde{\bm E}_{\text{col}}$, we also have $\sigma^2_{min}(\ER) \leq \sigma^2_{min}(\tilde{\bm E}_{\text{col}})$, which together with \eqref{eq:singvalineq} gives the desired inequality.

\section{Appendix B: Proof of Main Theorem}
\label{mainthmproof}
\subsection{Reformulation in lifted domain}
We now reformulate the recovery of $\hat f$ as a matrix recovery problem in the lifted domain. The matrices $\Txf$ and $\Tyf$ contain several copies of the weighted entries $k_x\widehat{f}[\bm k]$ and $k_y\widehat{f}[\bm k]$, respectively. We use $\omega(\bm k)$ to denote the set of locations $(\alpha_1,\alpha_2)$ in the matrix $\Txf$ or $\Tyf$ that contain the entry $k_x\widehat{f}[\bm k]$ or $k_y\widehat{f}[\bm k]$ (this set is the same in either case). 

We define the sampling matrices $\bm A_{\bm k} = \begin{bmatrix}\bm A_{1,\bm k}\\\bm A_{2,\bm k}\end{bmatrix} \in \mathbb{C}^{2|\Lambda_2|\times|\Lambda_1|}$, for each $\bm k = (k_1,k_2) \in \Gamma$, where 
\begin{eqnarray}
(\bm A_{i,\bm k})_{\bm \alpha} &=& \left\{\begin{array}{ccc}
\frac{k_i}{\|\bm k\|\sqrt{|\omega_i(\bm k)|}},&\mbox{ if }& \bm \alpha = (\alpha_1,\alpha_2) \in \omega(\bm k)\\
0 & \mbox{else}
\end{array}
\right.
\end{eqnarray}
for $i=1,2$. 
The matrices $\{\bm A_{\bm k}\}_{\bm k \in \Gamma}$ form an orthonormal basis for the space of matrices defined by the range of the matrix lifting $\mathcal{T}$; we will call any matrix in the range of $\mathcal{T}$ a \emph{structured matrix}. For any set of coefficients $\{\widehat g[\bm k]\}_{\bm k \in \Gamma}$ we can expand the structured matrix $\mathcal{T}(\widehat g)$ as 
\begin{equation}
\mathcal{T}(\widehat g) = \sum_{\bm k\in\Gamma} \widehat{g}[\bm k]\, \|\bm k\|\sqrt{|\omega_i(\bm k)|}\,\bm A_{\bm k}.
\end{equation}
We denote the projection operator corresponding to a single sampling location $\bm k$ by $
\mathcal A_{\mbf k}(\mbf X) =  \inner{\bm A_{\bm k},\bm X} \bm A_{\bm k}$. Since $\{\bm A_{\bm k}\}_{\bm k \in \Gamma}$ is an orthonormal basis, for any structured matrix $\mbf X$, we have $\sum_{k\in \Gamma} \mathcal A_{\mbf k}(\mbf X) = \mathcal A(\mbf X) = \mbf X$. Since $\bm A_{\bm k}$ is not the basis for a general $\bm X \in \mathbb{C}^{2|\Lambda_2|\times|\Lambda_1|}$, we also define the projection operator to the space orthogonal to the space of structured matrices by $
\mathcal A^{\perp}(\bm X) = (\mathcal I - \mathcal A)(\mbf X)$, where $\mathcal{I}$ is the identity operator.
In particular, the constraint $\mathcal A^{\perp}(\bm X) = \mathbf 0$ implies that $\bm X$ is a structured matrix. 

The recovery of $f$ from its partial Fourier samples $\widehat{f}[\bm k], \bm k\in \Omega,$ can thus be reformulated as the completion of a structured matrix $\bm X$ from its measurements $\mathcal{A}_{\bm k}, \bm k \in \Omega$.
Since the matrix is structured, we have $\mathcal A^{\perp}(\mbf X)= \mbf 0$. We thus reformulate \eqref{nucnorm} as the structured low-rank recovery problem:
\begin{equation}
\label{stlr}
\mbox{minimize}_{\bm X}~ \|\bm X\|_* ~\mbox{subject to}~
 \mathcal Q_{\Omega}(\bm X) = \mathcal  Q_{\Omega}(\Tf),
\end{equation}
where $\mathcal Q_{\Omega}$ that satisifies $\mbb E[\mathcal Q_{\Omega}] = \mathcal I$ is defined as:
\begin{align}
\label{Qomega}
\mathcal Q_{\Omega} &= \frac{|\Gamma|}{|\Omega|}\mathcal A_{\Omega} + \mathcal A^{\perp}
\end{align}
\subsection{Conditions for perfect recovery}
The tangent space $T$ of the matrix $\mathbf X$ is defined as $
T:= \{ \mathbf U\mathbf X_1^H+\mathbf X_2 \mathbf V^H: \mathbf X_1 \in \mathbb C ^{|\Lambda_2| \times R}, \mathbf X_2 \in \mathbb C^{|\Lambda_1|\times R} \}$
where $\mathbf X = \mathbf U \bs \Lambda \mathbf V^H$ is the singular value decomposition of $\mbf X$. The orthogonal complement of $T$ is denoted by $T^\perp$. 
We first show that if $\mathcal P_T \approx \mathcal P_T\mathcal Q_{\Omega}\mathcal P_T$, and if an approximate dual certificate that satisfies certain conditions exist, we obtain perfect recovery. 

\begin{lem}
\label{conditionslemma}
Consider a multiset $\Omega$ that contains $m$ random indices. Suppose the sampling operator $\mathcal Q_{\Omega}$ obeys 
\begin{equation}
{\|\mathcal P_T - \mathcal P_T\mathcal Q_{\Omega}\mathcal P_T\| \leq \frac{1}{2}}
\label{ptclose}\end{equation}
and there exists a dual certificate matrix $\bm W$ satisfying 
\begin{eqnarray}
\label{cond2}
\mathcal Q_{\Omega}^{\perp} (\bm W) &=& 0\\
\label{cond3}
\|\mathcal P_T (\bm W - \bm U \bm V^*)\|_F &\leq& \frac{1}{6n}\\ 
\label{cond4}
\|\mathcal P_T^{\perp}(\bm W)\| &\leq& \frac{1}{2}.
\end{eqnarray}
Then, $\mathcal{T}(\hat f)$ is the unique solution to \eqref{stlr}, where $n= |\Gamma|$ and $m=|\Omega|$.
  \label{lemQ}
\end{lem}

See Section \ref{condlemmaproof} of supplementary material for proof. Equation \eqref{ptclose} suggests that $\mathcal Q_{\Omega} \approx \mathcal I$ on the tangent space. The conditions \eqref{cond2}, \eqref{cond3}, and \eqref{cond4} indicates the existence of a $\mbf W$, which approximates the exact dual certificate $\mbf U\mbf V^*$. The above lemma is in line with \cite[lemma 1]{chen2014robust}, with the exception of the third condition, indicated by \eqref{cond3}. To satisfy \eqref{ptclose}, we bound the deviation of $\mathcal P_T \mathcal Q_{\Omega} \mathcal P_T$ from $\mathcal P_T$ in the following lemma.
\begin{lem}
\label{ptcloselem}
Suppose \eqref{coherence2} holds. Then we have 
\begin{equation}
\|\mathcal P_{T} - \mathcal P_T \mathcal Q_{\Omega} \mathcal P_T\| \leq \epsilon \leq \frac{1}{2} 
\end{equation}
with probability exceeding $1-n^{-4}$, provided that $m > c_1 \rho R \,c_s\, \log(n)$.
\end{lem}
We prove this using \cite[Theorem 1.6]{tropp2012user}. (See Section \ref{ptcloselemproof} of supplementary material)
\subsection{Construction of the approximate dual certificate $\mbf W$}
We will now use the golfing scheme of \cite{gross,chen2014robust} to construct an approximate dual certificate $\bm W$, which satisfies \eqref{cond2}, \eqref{cond3}, and \eqref{cond4}. In particular, we generate $j_0$ independent random sampling sets $\Omega_i; 1\leq i\leq j_0$, each containing $\tilde m = m/j_0$ samples corresponding to sampling with replacement. 
We start with $\bm F_0 = \bm U \bm V^*$, and follow the following steps:
\begin{enumerate}
\label{Our}
\item {$~\bm F_0 = \bm U \bm V^*$ and set $j_0=3\log_{\frac{1}{\epsilon}} n $.}
\item {$~\forall i (1 \leq i \leq j_0), \bm F_i =  \mathcal{P_T}(\mathcal I-\mathcal Q_{\Omega_i})\mathcal{P}_T(\bm F_{i-1}) $}
\item {$~\bm W = \sum_{j=1}^{j_0} \mathcal Q_{\Omega_i} \bm F_{j-1} $ }
\end{enumerate} 

Step 3 ensures that $\mbf W$ satisfies \eqref{cond2} since each term $\mbf W_i =\mathcal Q_{\Omega_i} \mbf F_{j-1}$ satisfies $\mathcal Q_{\Omega}^{\perp} (\bm W_i) = 0$. The recursive construction also satisfies \eqref{cond3}. 
In particular,  
%
\begin{eqnarray*}
\|\mathcal P_T( \mbf W- \mbf U\mbf V^*)\|_F &=&\|\mathcal P_T \mbf F_{j_0}\|_F \\
&\leq& \epsilon^{j_0} \|\mbf F_0\|_F = \epsilon^{j_0} \sqrt R \leq \epsilon^{j_0} n
\end{eqnarray*}

Now we focus on showing that $\mbf W$ satisfies \eqref{cond4}. Note that if $j_0$ is chosen as $3\log_{\frac{1}{\epsilon}} n$, assuming $n>6$, we have $\left(\epsilon \right)^{j_0} \; n < \frac{1}{6~ n}$. 


\begin{lem}
\label{lemma9}
For any matrix $\mbf M$, there exists some numerical constant $c_2$ such that 
\begin{equation}
\left\|\left(\mathcal I-\mathcal Q_{\Omega}\right) (\mbf M)\right\| \leq c_2 \sqrt{\frac{n\log n}{m}}\|\mbf M\|_{\mathcal A,2} + {\frac{c_2 n\log n}{m}}\|\mbf M\|_{\mathcal A,\infty},
\end{equation}
with probability at least $1- n^{-10}$. Here, 
\begin{eqnarray}
\|\mbf M\|_{\mathcal A,\infty} = \max_{\mbf k\in\Gamma} \left|\frac{\inner{\mathbf A_{\mbf k},\mbf M}}{\abs{\omega_k}}\right|\\
\|\mbf M\|_{\mathcal A,2} = \sqrt{\sum_{\mbf k \in \Gamma} \frac{\abs{\inner{\mathbf A_{\mbf k},\mbf M}}^2}{\abs{\omega_k}}}
\end{eqnarray}
\end{lem}
See Section \ref{lemma9proof} of supplementary material for proof.
\begin{lem}
\label{lemma10}
Assume that there exists a constant $\mu_5$ such that $\omega_{\mbf k} \|\mathcal P_T(\mbf A_{\mbf k})\|_{\mathcal A,2} \leq \frac{\mu_5 R}{n}$. For any matrix $\mbf M$, we have 
\begin{multline}\nonumber
\|\mathcal P_T [(\mathcal I-\mathcal Q_{\Omega}) (\mbf M)]\|_{\mathcal A,2} \\\nonumber \leq c_3 \sqrt{\frac{\mu_5\, R\, \log n}{m}} ~\bkt{\|\mbf M\|_{\mathcal A,2} + \sqrt{\frac{n\, \log n}{m}}~ \|\mbf M\|_{\mathcal A,\infty}},
\end{multline}
with probability at least $1- n^{-10}$. \end{lem}
See Section \ref{lemma10proof} of the Supplementary Materials for proof.


\begin{lem}
\label{lemma11}
For any matrix $\mbf M \in T$, there exists some numerical constant $c_4$, such that 
\begin{multline}
\|\mathcal P_T[ (\mathcal I-\mathcal Q_{\Omega} )(\mbf M )] \|_{\mathcal A,\infty} 
\\ \leq c_4 \sqrt{\frac{ \rho c_s R\log n}{m}} \sqrt{\frac{\rho c_s R}{n}}\|\mbf M\|_{\mathcal A,2} + \frac{c_4 \rho c_s R\log n}{m}\|\mbf M\|_{\mathcal A,\infty},
\end{multline}
with probability at least $1- n^{-10}$. \end{lem}
See Section \ref{lemma11proof} of the Supplementary Materials for proof.

From the golfing scheme, we have $\|\mathcal P_{T^{\perp}}(\mbf W)\| \leq \sum_{j=1}^{j_0}  \|\mathcal P_{T^{\perp}}\mathcal Q_{\Omega_i} \mathcal P_{T}\bm F_{j-1} \|$. Using lemma \ref{lemma9} 
and substituting from lemma \ref{lemma10} and lemma \ref{lemma11}, we have
\begin{multline}\nonumber
\|\mathcal P_{T^{\perp}}\mathcal Q_{\Omega_i} \bm F_{j-1} \| \\
\label{wbound}
\leq  \bkt{\frac{1}{2}}^{j_0-1}c_2\brac{ \sqrt{\frac{n\log n}{\tilde m}} ~\|\bm F_{0}\|_{\mathcal A,2} + {\frac{n\log n}{\tilde m}}~ \|\bm F_{0}\|_{\mathcal A,\infty}}
\end{multline}
 The last inequality holds if $\tilde m = m/j_0 \gg \max\bkt{\mu_5,\rho c_s} R \log n $. Substituting for $j_0=3\log_{\frac{1}{\epsilon}}(n)$ assumed in the golfing scheme, we require $
m \gg c_6~\max\bkt{\mu_5,\rho c_s} R \log^2 n$  to satisfy the above inquality. See Section \ref{ptperpbounddetails} of the Supplementary Materials for details. We will now present the lemmas bounding $\|\bm F_{0}\|_{\mathcal A,2}$ and $\|\bm F_{0}\|_{\mathcal A,\infty}$, where $\mbf F_0=\mbf U\mbf V^*$.

\begin{lem}
\label{lem20}
With the incoherence measure $\rho$, one can bound
\begin{eqnarray}
\label{uvinfty}
\|\mbf U\mbf V^*\|_{\mathcal A,\infty} &\leq& \frac{\rho ~ c_s R}{n}\\
\label{uvl2}
\left\|\mbf U\mbf V^*\right\|_{\mathcal A,2}^2 &\leq& \frac{c_7\mu_3 ~c_s \log^2 (n) R }{n}\\
\label{uvl3}
\left\|\mathcal P_T\bkt{\sqrt{\omega_{\alpha}} ~\mathbf A_{\alpha}}\right\|_{\mathcal A,2}^2
 &\leq& \frac{c_7 \mu_3 c_s \log^2(n) R}{n}, \forall \mbf \alpha\in\Gamma
 \end{eqnarray}
 for $\mu_3 = 3\rho$ and $c_7$ is some constant.
\end{lem}
See Section \ref{lem20proof} of the Supplementary Materials for proof. From \eqref{uvl3}, we see that the constant $\mu_5 $ in lemma \ref{lemma11} can be chosen as 
$\mu_5 = c_7~ \mu_3~ c_s ~\log^2(n)$
such that  $\omega_{\mbf k} \|\mathcal P_T(\mbf A_{\mbf k})\|_{\mathcal A,2} \leq \frac{\mu_5 R}{n}$.
 Substituting for $\mu_5$, we observe that the dominant term has its dependence on $\log^4(n)$. Thus, $\|\mathcal P_{T^{\perp}}\mathcal Q_{\Omega_i} \bm F_{j-1} \|  < 1/2$ if 
$m> c_6  c_7\, c_s\, \bkt{3\rho} \,R\,\log^4(n).$

\bibliographystyle{IEEEtran}
\bibliography{IEEEabrv,refs}

\onecolumn
{\Huge \begin{center}
Supplementary Materials for: Convex recovery of continuous domain piecewise constant images from non-uniform Fourier samples
\end{center}}
\vspace{3pt}

\begin{center}
Greg Ongie,\textit{~Member,~IEEE}, Sampurna Biswas, \textit{~Student Member,~IEEE}, Mathews Jacob, \textit{~Senior Member, IEEE}
\end{center}
\vspace{13pt}
\begin{center}
This report elaborates the technical details of the paper titled, ``Convex recovery of continuous domain piecewise constant images from non-uniform Fourier sample''.
\end{center}
\newpage
\setcounter{section}{0}

\section{Proof of Theorem \ref{geometry}}\label{sec:supp_moitra}

 The proof we give is multi-dimensional generalization of the proof of  \cite{moitra2015super}. We will make use the following lemma from \cite{moitra2015super}:
\begin{lem} 
\label{minorant}
There is an entire function $c_E(t)$ whose Fourier transform is supported in the interval $-\Delta,\Delta$, which satisfies $c_E(t) < I_E(t)$---the indicator function of the interval $E=[-n/2,n/2]$.
\begin{equation}
\int_{-\infty}^{\infty} \bkt{I_E(t) - c_E(t)} dt = \frac{1}{\Delta}
\end{equation}
\end{lem}
The above function $c_E(t)$ is known as the Beurling-Selberg minorant of $I_E(t)$. Note that $\widehat{I_E}(0) = n$ and hence $\widehat{c_E}(0) = n-1/\Delta$, where $\widehat{I_E}$ is the Fourier transform of $I_E$.

We now give the proof of Theorem 10:
\begin{proof}
We note that $\lambda_{min}\bkt{\ER^* \ER} =  \min_{\|\mbf u\|=1} \|\ER \mbf u\|^2$. From the definition of $\ER$, we have \\ $\ER \mbf u = \frac{1}{\sqrt{|\Lambda_1|}}\sum_{i=1}^{R} u_i \exp\bkt{j 2\pi \mbf k\cdot \bm r_i}; \mbf k \in \Lambda_1$.  
We consider the continuous domain function 
\begin{equation}
v(\mbf f) =  \frac{1}{\sqrt{|\Lambda_1|}}\sum_{i=1}^{R} \mbf u_i \exp\bkt{j 2\pi \mbf f\cdot \mbf r_i}; \mbf f = (f_1,f_2) \in \mathbb R^2
\end{equation}
 and rewrite the discrete summation in $\|\ER \mbf u\|^2$ as the integral
 \begin{equation}
 \|\ER \mbf u\|^2 =  \frac{1}{{|\Lambda_1|}}\int_{-\infty}^{\infty} \abs{v(\mbf f)}^2\,I_E(f_1)I_E( f_2)\,h(f_1)h(f_2)\, d\mbf f\\
\end{equation}
where $h(f) = \sum_{m=-\infty}^{\infty} \exp\bkt{j 2\pi f m}$ is the Dirac comb function and $I_E(f)$ is the indicator function of the region \\ $[- \sqrt{\frac{|\Lambda_1|}{2}}, \sqrt{\frac{|\Lambda_1|}{2}}]$. Minorizing $I_E$ by $c_E$, specified by Lemma \ref{minorant}, we obtain
\begin{eqnarray}
 \|\ER \mbf u\|^2 &\geq&  \frac{1}{{|\Lambda_1|}}\int_{-\infty}^{\infty} \abs{v(\mbf f)}^2\,c_E(f_1)c_E( f_2)\,h(f_1)h(f_2)\, d\mbf f\\
  &=&   \frac{1}{{|\Lambda_1|}}\sum_{i=1}^{R}\sum_{j=1}^{R}\mbf u_i \mbf u_j^*  \int_{-\infty}^{\infty}\exp\bkt{j 2\pi f (\mbf r_i-\mbf r_j)} \,c_E(f_1)c_E( f_2)\,h(f_1)h(f_2)\, df_1 df_2\\
    &=&  \frac{1}{{|\Lambda_1|}}\sum_{i=1}^{R}\sum_{j=1}^{R}\sum_{m=-\infty}^{\infty}\sum_{m'=-\infty}^{\infty}\mbf u_i\, \mbf u_j^* ~\widehat{c_E}(x_i - x_j-m)\widehat{c_E}(y_i - y_j-m')
\end{eqnarray}
In the last step, we used the expression of the Dirac comb function. Since the Fourier transform of $c_E$ is supported within $[-\Delta,\Delta]$, the terms $\widehat{c_E}(x_i - x_j-m)$ and $\widehat{c_E}(x_i - x_j-m)$ are non-zero only if $x_i=x_j; m=0$ and  $y_i=y_j; m=0$, since $|x_i-x_j| > \Delta, |y_i-y_j| > \Delta; i\neq j$. Thus, we have 
\begin{eqnarray}
 \|\ER \mbf u\|^2 &\geq&  \frac{1}{{|\Lambda_1|}}\sum_{i=1}^{R}\sum_{j=1}^{R}\mbf u_i\, \mbf u_j^* ~\underbrace{\widehat{c_E}(x_i - x_j)\widehat{c_E}(y_i - y_j)}_{\delta(\mbf r-\mbf r')}\\
    &=&  \frac{1}{{|\Lambda_1|}}\widehat{c_E}(0)^2~\|\mbf u\|^2 =  \bkt{1-\frac{1}{ \sqrt{|\Lambda_1|}~\Delta}}^2
\end{eqnarray}

\end{proof}
\section{The BKK bound (Proof of Lemma 11)}
\label{SMbkk}
The BKK bound in enumerative algebraic geometry (see, e.g., \cite{li1996bkk}) is a well-known result that relates maximum number of isolated solutions of a system of polynomials to their coefficient support sets. Specifically, the BKK bound is typically stated in terms of \emph{Laurent polynomials}, i.e., functions of the form $q(t,s) = \sum_{k=-\infty}^\infty \sum_{\ell = -\infty}^\infty t^k s^\ell$ with $t,s\in\mathbb{C}$, and only finitely many non-zero coefficients $c_{k,\ell} \in \mathbb{C}$. Since a trigonometric polynomial is the restriction of a Laurent polynomial to the complex unit torus $\{(t,s) = (e^{j2\pi x},e^{j2\pi y}): x,y\in[0,1)\}$, the result also holds for trigonometric polynomials, which we state below:
\begin{thm*}[BKK Bound]
\label{bkk}
Let $\mu_1$ and $\mu_2$ be trigonometric polynomials with coefficient supports $\Omega_1$ and $\Omega_2$, and let $P_1 = \text{conv}(\Omega_1)$ and $P_2 = \text{conv}(\Omega_1)$, where $\text{conv}(\cdot)$ denotes the convex hull of a set in $\mathbb{Z}^2$ treated as a subset of $\mathbb{R}^2$. The number of isolated solutions of the system $\mu_1(\bm r) = \mu_2(\bm r ) =0$ for $\bm r \in [0,1)^2$ is at most
 \[
 \mathcal{M}(\Omega_1,\Omega_2) := area(P_1 + P_2) - area(P_1) - area(P_2),
 \]
where $area(\cdot)$ denotes the usual Euclidean area. In particular, if $\mu_1$ is irreducible, and the common isolated zeros of $\mu_2$ and $\mu_1$ are greater than $\mathcal{M}(\Omega_1,\Omega_2)$, then $\mu_1$ must divide $\mu_2$.
\end{thm*}
When $\Lambda_0$ and $\Lambda_1$ are rectangular index sets satisfying $\Lambda_0 \subset \Lambda_1$, a straightforward computation reveals that
\[
R < \mathcal{M}(\Lambda_0,\Lambda_1) < R + |\Lambda_0|,
\]
which establishes the bound in Lemma 11.
\section{Proofs of results in Appendix B}

\subsection{Proof of Lemma \ref{conditionslemma}}
\label{condlemmaproof}
\begin{proof}
Let $\bm X$ be the unique minimizer for the convex optimization problem \eqref{stlr} and $\bm H$ be a perturbation of $\bm X$. To prove exact recovery of \eqref{stlr}, it suffices to show the existence of an exact dual certificate $\bm W$, s.t, $\|\bm X+\bm H\|_* > \|\bm X\|_*+ \inner{\bm W, \bm H}$. With an approach similar to \cite{gross}, we now show that the existence of an approximate dual certificate will guarantee unique recovery. We separately consider the two cases based on the relative energies of $\mbf H_T$ and $\mbf H_T^{\perp}$, where $\mbf H=\mbf H_T+\mbf H_T^{\perp}$ and $T$ denotes the tangent space.

\subsubsection{\underline{Case 1: $\| \bm H_T\|_F > 2 n\, \|\bm H_T^{\perp}\|_F$}}
We will show that $\bm X + \bm H$ is infeasible (i.e, $\| \mathcal{Q}_\Omega \bm H \|_F >0$), if $\mathcal Q_{\Omega} \approx \mathcal I$ on the tangent space (i.e, \eqref{ptclose} is satisfied). We have 
\begin{eqnarray}
\label{infeasibility}
\| \mathcal{Q}_\Omega \bkt{\bm H_T + \bm H_T^{\perp}}  \|_F \geq \| \mathcal{Q}_\Omega \bm H_T \|_F -\| \mathcal{Q}_\Omega \bm H_T^{\perp} \|_F
\end{eqnarray}
We upper bound the second term as $\| \mathcal Q_{\Omega} \bm H_T^{\perp} \|_F \leq \|\mathcal Q_{\Omega}\| ~\|\bm H_T^{\perp}\|_F$. 
By definition, we have $\|\mathcal Q_{\Omega}\|  = \|~\frac{n}{m}~ \mathcal A_{\Omega} + \mathcal A^{\perp}\|$, from which we obtain $\left\|\mathcal Q_{\Omega}\right\|  = \left\|\frac{n}m \mathcal A_{\Omega} + \mathcal A^{\perp}\right\| 
\leq \|\frac{n}m \bkt{\mathcal A + \mathcal A^{\perp}}\| = \frac{n}m < n$. 
We omit the m term in the denominator to remove the dependence on m. Hence, we have $\| \mathcal Q_{\Omega} \bm H_T^{\perp} \|_F \leq n  \|\bm H_T^{\perp}\|_F$.
We now lower bound the first term in \eqref{infeasibility}:
\begin{align*}
\label{line2new}
\| \mathcal Q_{\Omega}  \bm H_T \|_F^2 
&\geq\| \bm H_T\|_F^2 \bkt{1- \underbrace{\|\mathcal P_T - \mathcal P_T\mathcal Q_{\Omega}\mathcal P_T\|}_{\leq \frac 1 2}} 
\end{align*}
Since we assumed that $\| \bm H_T\|_F > n\, \|\bm H_T^{\perp}\|_F$, we have $
\| \mathcal{Q}_\Omega \bm H \|_F \geq \bkt{\sqrt 2 -1}\, n~\|\bm H_T^{\perp}\|_F > 0$, implying that such an $\mbf H$ is infeasible.


\subsubsection{\underline{Case 2: $\| \bm H_T\|_F \leq  2 n\, \|\bm H_T^{\perp}\|_F $}} 
\label{case2}
We now show that if $\bm H_T$ is small and $\bm X +\bm H$ is feasible, then the nuclear norm of $\bm X+\bm H$ is larger than $\bm X$. Since $\bm X +\bm H$ is feasible, we have $\mathcal Q_{\Omega}(\bm H) = 0.$
If $\mbf X = \bm U \bm \Sigma \bm V^*$ represents the singular value decomposition of $\bm X$, the subgradient of $\|\bm X\|_*$ is parametrized as $\bm U \bm V ^* + \bm Z_0; ~\bm Z_0 \in T^{\perp}$. By definition of subgradient, we have $
\| \bm X + \bm H \|_* \geq  \| \bm X \| _* + \inner{\bm U\bm V^*, \bm H} + \inner{\bm Z_0,\bm H}$. We consider a $\bm W$ in the range space of $\mathcal Q_{\Omega}$ that satisfies \eqref{cond2} $\mathcal Q_{\Omega}^{\perp} (\bm W) = 0$, 
where $\mathcal Q_{\Omega}^{\perp}$ is the projection $\mathcal Q_{\Omega}^{\perp} = \mathcal A - \mathcal A_{\Omega}$. Using this $\bm W$, we rewrite the above relation as
\begin{eqnarray}\nonumber
\label{secondeqn}
\| \bm X + \bm H \|_* &\geq&  \| \bm X \| _* + \underbrace{\inner{\bm W, \bm H}}_0 + \underbrace{\inner{\bm Z_0,\bm H}}_{\color{black}\leq \|\bm H_T^{\perp}\|_*} + \inner{\bm U\bm V^* -\bm W,\bm H} \\
&\geq& \| \bm X \| _* + \|\bm H_T^{\perp}\|_*- \inner{\bm W -\bm U\bm V^*,\bm H} 
\end{eqnarray}
The second term in \eqref{secondeqn} vanishes since $\bm W$ lives in the range of $\mathcal Q_{\Omega}$ (i.e., $\mathcal Q_{\Omega}^{\perp}(\bm W)=0$) while $\bm H$ lives in the kernel of  $\mathcal Q_{\Omega}$ (i.e., since $\mathcal Q_{\Omega}(\bm H)=0$). Since $\bm Z_0 \in T^{\perp}$ and $\|\bm Z_0\| \leq 1$, the third term is less than or equal to $\|P_T^{\perp} \bm H\|_*$. We now focus on the last term. If $\bm W$ satisfies \eqref{cond3} and \eqref{cond4}, we have 
\begin{align*}
\inner{\bm W -\bm U \bm V ^* ,\bm H}
&\leq \|\mathcal{P}_T (\bm W -\bm U \bm V ^*) \|_F\; \| \bm H_T\|_F + \| \mathcal{P}_T^{\perp} \bm W\| \;\| \bm H_T^{\perp} \|_*\\
&\leq  \frac{1}{6 \;n} ~\| \bm H_T\|_F ~ + ~\frac{1}{2}\, \| \bm H_T^{\perp} \|_F
\end{align*}
Substituting in \eqref{secondeqn} and using $\|\mbf M\|_* \geq \|\mbf M\|_F$, we have 
\begin{eqnarray}
\| \bm X + \bm H \|_* 
&\geq&  \| \bm X \| _* + \left(\frac{1}{2} - \frac{1}{3} \right) \|\bm H_T^{\perp}\|_F 
\label{one6}
\end{eqnarray}
We used $ \|\bm H_T\|_F \leq 2 n\, \|\bm H_T^{\perp}\|_F $ in the above. The above inequality implies that $\| \bm X + \bm H \|_* \geq  \| \bm X \| _*$ and hence $\bm X$ is the unique minimizer. 
\end{proof}

\subsection{Proof of Lemma \ref{ptcloselem} }
\label{ptcloselemproof}

\begin{proof}
Substituting the definition of $\mathcal Q_{\Omega}$ in \eqref{ptclose}, we see that 
\begin{equation}
\mathcal P_{T}\bkt{\mathcal I - \mathcal Q_{\Omega}} \mathcal P_T = \mathcal P_{T}\bkt{\mathcal A - \frac{n}{m} \mathcal A_{\Omega}} \mathcal P_T
\end{equation}
We set $\mathcal Z_{\bm k} = \frac{n}{m}~\mathcal P_T \mathcal A_{\bm k}\mathcal P_T$ with $\mbb E[\mathcal Z_{\bm k}] = \frac{1}{m} ~\mathcal P_T \mathcal A \mathcal P_T$, we obtain 
\begin{equation}
\mathcal P_{T}\mathcal A\mathcal P_T -\frac{n}{m} \mathcal P_T \mathcal A_{\Omega} \mathcal P_T = \sum_{\bm k \in \Omega} \underbrace{{\mbb E[\mathcal Z_{\bm k}] -\mathcal Z_{\bm k}}}_{\mathcal S_{\bm k}} \nonumber
\end{equation}
We apply the operator Bernstein's inequality  \cite[Theorem 1.6]{tropp2012user} to determine $\mathbb P\left(\left\|\sum_{\bm k \in \Omega} \mathcal S_{\bm k}\right\|\geq \epsilon\right)$. Using steps similar to  \cite[Lemma 3]{chen2014robust}, we obtain the upper bounds $ \left\| \mbb{E} \left[\mathcal S_k^2\right]~\right\| \leq \frac{4 \rho Rc_s}{m}$ and $\| \mathcal S_k \| \leq  \frac{4 \rho Rc_s}{m}$. 
We now apply the matrix Bernstein's inequality \cite[Theorem 1.6]{tropp2012user} to obtain
\begin{eqnarray}
\mathbb P\left(\left\|\sum_{\bm k \in \Omega} \mathcal S_{\bm k}\right\|\geq \frac{1}{2}\right) &\leq n~\exp\left(\frac{-1/8}{\frac{4 \rho Rc_s}{m}\left(1+1/6\right)}\right)
\end{eqnarray}
We desire $\mathbb P\left(\left\|\sum_{\bm k \in \Omega} \mathcal S_{\bm k}\right\|\geq \epsilon \right) < (n)^{-b}$. Setting these values in the above inequality, we obtain \\ \textcolor{black}{$
(n)^{-b} \geq  n~\exp\left(\frac{-3 m}{112~\rho R c_s}\right)$.} Taking log of both sides and simplifying, we obtain \textcolor{black}{ 
\begin{equation}
m \geq \underbrace{(b+1) \frac{112}{3}}_{c_1} ~\rho R \,c_s ~\log(n)
 \end{equation}}
\end{proof}

\subsection{Proof of Lemma \ref{lemma9}}
\label{lemma9proof}
\begin{proof}
We define 
\begin{equation}
\mathcal S_{\mbf k} = \frac{n}{m} \mathcal A_{\mbf k} - \frac{1}{m} \mathcal A,
\end{equation}
which satisfies $\mbb E(\mbf S_k) = 0$ and $\left\|\mathcal I-\mathcal Q_{\Omega}\right\|= \|\sum_{\mbf k\in \Omega} \mathcal S_k\|$. To bound the right hand side using operator Bernstein's inequality, we require the bounds $\|\mathcal S_k(\mbf M)\| \leq B$ and $\sigma^2 = \max \left\{\left\|\sum_k~\mbb E[\mathcal S_k\mathcal S_k^*]\right\|,\left\|\sum_k~\mbb E[\mathcal S_k^*\mathcal S_k]\right\|  \right\}$.
We first consider 
\begin{eqnarray*}
\mathcal S_{\mbf k}^* \mathcal S_{\mbf k}(\mbf M)
&\leq& \bkt{\frac{n}{m}}^2 \mathcal A_{\mbf k}^*\mathcal A_{\mbf k}(\mbf M)= \bkt{\frac{n}{m}}^2 \abs{\inner{\mbf A_{\mbf k},\mbf M}}^2 \mbf A_{\mbf k}^T \mbf A_{\mbf k}\\
&\leq& \bkt{\frac{n}{m}}^2 \frac{\abs{\inner{\mbf A_{\mbf k},\mbf M}}^2}{\omega_{\mbf k}} \mbf I_{|\Lambda_1|\times |\Lambda_1|}
\end{eqnarray*}
which gives $\left\|\mbb E\left[\sum_{k\in \Omega}~\mathcal S_k^*\mathcal S_k\right]\right\| = \frac{m}{n}\left\|\sum_{k\in \Gamma}~\mathcal S_k^*\mathcal S_k\right\|
\leq \frac{n}{m} \|\bm M\|_{\mathcal A,2}^2$. Similarly, we have 
\begin{eqnarray*}
\mathcal S_{\mbf k} \mathcal S_{\mbf k}^*(\mbf M)&\leq& \bkt{\frac{n}{m}}^2 \abs{\inner{\mbf A_{\mbf k},\mbf M}}^2 \mbf A_{\mbf k}\mbf A_{\mbf k}^T \\
&\leq& \bkt{\frac{n}{m}}^2 \frac{\abs{\inner{\mbf A_{\mbf k},\mbf M}}^2}{\omega_{\mbf k}}~~\bkt{ \frac{\mbf k}{\|\mbf k\|} \frac{\mbf k^T}{\|\mbf k\|}} \bigotimes \mbf I_{|\Lambda_2|\times |\Lambda_2|},
\end{eqnarray*}
where $\mbf k = (k_1,k_2)^T$, which also gives the bound $\left\|\mbb E\left[\sum_{k\in \Omega}~\mathcal S_k\mathcal S_k^*\right]\right\| \leq  \frac{n}{m} \|\bm M\|_{\mathcal A,2}^2$. Here, $\bigotimes$ denotes the Kroneker product. 
Similar to the arguments in \cite{chen2014robust}, we have $\|\mathcal S_k \mbf M\| \leq \frac{2n}{m} \|\mbf M\|_{\mathcal A,\infty}$. Combining these terms into \cite[Lemma 11]{chen2014robust}, the result is proved.
\end{proof}

\subsection{Proof of Lemma \ref{lemma10}}
\label{lemma10proof}

\begin{proof}
We note that 
\begin{eqnarray}
\|\mathcal P_T[(\mathcal I-\mathcal Q_{\Omega})(\mbf M)]\|_{\mathcal A,2}^2 = \sum_{\mbf k \in \Gamma} \frac{\abs{\inner{\bm A_{\mbf k},\mathcal P_T[(\mathcal I-\mathcal Q_{\Omega})(\mbf M)]}}^2}{\abs{\omega_k}}\end{eqnarray}
We assume that $\mathcal A_{\Omega} = \sum_{i=1}^m \mathcal A_{{\bm \alpha}_i}$, where ${\bm \alpha}_i$ are independent indices picked at random. Correspondingly, we consider vectors $\mbf z_{{\bm \alpha}_i}$ of length $n= |\Gamma|$, whose entries are specified by
\begin{equation}
\label{zalpha}
z_{{\bm \alpha}}\bkt{\mbf k}=\frac{1}{\sqrt{\omega_{\mbf k}}}\inner{\mbf A_k,\mathcal P_T\left [\frac{n}{m}(\mathcal A_{\bm \alpha}-\mathcal A) (\mbf M)\right ]}
\end{equation}
Note that the desired bound $\left\|\mathcal P_T[(\mathcal I-\mathcal Q_{\Omega}) (\mbf M)]\right\|_{\mathcal A,2}=\left \| \sum_{i=1}^m \mbf z_{{\bm \alpha}_i} \right \|_{2}$. We have $\mbb E(\mbf z_{{\bm \alpha}}) = 0$. We proceed as \cite[Lemma 5]{chen2014robust} and from the definition of $\mu_5$, we have,
$\|\mbf z_{{\bm \alpha}_i}\|_2 \leq 2 \sqrt{\frac{n \mu_5 R}{m^2}} \|\mbf M\|_{\mathcal A,\infty}$ and as $\mbf z_{{\bm \alpha}_i}$'s are vectors, $$\left \|\mathbb E \left[\sum_{i=1}^m \mbf z_{{\bm \alpha}_i} \mbf z_{{\bm \alpha}_i}^*\right] \right \| = \left \|\mathbb E \left[\sum_{i=1}^m \mbf z_{{\bm \alpha}_i}^* \mbf z_{{\bm \alpha}_i}\right] \right \| \leq \frac{4 \mu_5 R}{m}\|\mbf M\|^2_{\mathcal A,2}$$ 
Substituting these bounds in the operator Bernstein's inequality\cite[Lemma 11]{chen2014robust}, the result is proved.
\end{proof}
\subsection{Incoherence between two sampling bases}
The proof of Lemma \ref{lemma11} relies on the bound on $\abs{\inner{\bm A_{\bm \beta},\mathcal P_T \bm A_{\bm  \alpha}}}$.
We introduce the following lemma to establish the incoherence between two sampling bases:
%

\begin{lem}
\label{prop:2bases}
\textcolor{black}{Under the incoherence conditions of proposition 3 and definition 9 of the main text,}
\[ \label{eq:2bases}
\abs{\inner{\bm A_{\bm \beta},\mathcal P_T \bm A_{ \bm \alpha}}} \leq 3\sqrt{\frac{\omega_{\beta}}{\omega_{{\bm \alpha}}}}\frac{\rho}{|\Gamma|}
\]
\end{lem}

We will now bound $\abs{\inner{\bm A_{\beta},\mathcal P_T \bm A_{{\bm \alpha}}}}$. We see that 
\begin{equation}\label{sumterms}
\abs{\inner{\bm A_{\beta},\mathcal P_T \bm A_{{\bm \alpha}}}} \leq \abs{\inner{\bm A_{\beta},\mbf U\mbf U^* \bm A_{{\bm \alpha}}}}+
\abs{\inner{\bm A_{\beta},\bm A_{{\bm \alpha}}\mbf V\mbf V^*}}+ \abs{\inner{\bm A_{\beta},\mbf U\mbf U^* \bm A_{{\bm \alpha}}\mbf V\mbf V^*}}
\end{equation}
We will now bound each of the terms in the right hand side. We observe that $\mbf A_{\beta}$ has $\omega_{\beta}$ entries of magnitude $\frac{k_{\beta,1}}{\|\mbf k_{\beta}\|\sqrt{2\omega_{\beta}}}$ and $\omega_{\beta}$ entries of magnitude $\frac{k_{\beta,2}}{\|\mbf k_{\beta}\|\sqrt{2\omega_{\beta}}}$. Hence, we have $\|\bm A_{\beta}\|_{\ell_1} \leq \sqrt{\omega_{\beta}}$. We consider
\begin{eqnarray}
\abs{\inner{\bm A_{\beta},\mbf U\mbf U^* \bm A_{{\bm \alpha}}}} &\leq& \|\bm A_{\beta}\|_{\ell_1}\,\|\mbf U\mbf U^* \bm A_{{\bm \alpha}}\|_{\infty}\\
&\leq&  \sqrt{\omega_{\beta}}\,\max_{i,j}\abs{\mbf U\mbf U^* \bm A_{{\bm \alpha}}}\\
&\leq&  \sqrt{\frac{\omega_{\beta}}{\omega_{{\bm \alpha}}}}\,\max_{i,j}\abs{\mbf U\mbf U^* }
\end{eqnarray}
We now bound the entries of $\mbf U\mbf U^*$:
\begin{eqnarray*}
|(\mbf U \mbf U^*)_{\mbf k,l}| &\leq& \abs{(\mbf e_{k}^T \mbf \EL)(\mbf \EL^*\mbf \EL)^{-1}(\mbf \EL^* \mbf e_l)}\\
&\leq&\|\mbf e_{k}^T \mbf \EL \|_{F}^2~\|\mbf \EL^*\mbf \EL\|^{-1}\\
& \leq& \frac{\rho\, c_s\, R}n,
\end{eqnarray*}
which gives $\abs{\inner{\bm A_{\beta},\mbf U\mbf U^* \bm A_{{\bm \alpha}}}} \leq  \sqrt{\frac{\omega_{\beta}}{\omega_{{\bm \alpha}}}} \frac{\rho\, c_s\, R}n.$ Proceeding along the same lines, we obtain $\abs{\inner{\bm A_{\beta},\bm A_{{\bm \alpha}}\mbf V\mbf V^* }} \leq  \sqrt{\frac{\omega_{\beta}}{\omega_{{\bm \alpha}}}} \frac{\rho\, c_s\, R}n$ and $\abs{\inner{\mbf U\mbf U^*\bm A_{{\bm \alpha}}\mbf V\mbf V^*,\bm A_{{\bm \alpha}} }} \leq  \sqrt{\frac{\omega_{\beta}}{\omega_{{\bm \alpha}}}} \frac{\rho\, c_s\, R}n$. Substituting these expressions into \eqref{sumterms}, we prove the Lemma.
\subsection{Proof of Lemma \ref{lemma11}}
\label{lemma11proof}

\begin{proof}
Proceeding with the definition of $\mbf z_{{\bm \alpha}}$ in \eqref{zalpha}, we observe that 
\begin{align}
\left|\mbf z_{{\bm \alpha}_i}^k \right | \leq 2 \max_{\substack{k}} \frac{|\inner{\mbf A_k, \frac{n}{m} \mathcal P_T(\mbf A_{\bm \alpha})\inner{\mbf A_{\bm \alpha},\mbf M}}|}{\omega_k} 
\label{52}
\end{align}
Using Lemma \ref{prop:2bases}, we obtain
\begin{align}
\left | \inner{\mbf A_b, \mathcal P_T \mbf A_a} \right | \leq \sqrt{\frac{\omega_b}{\omega_a}}\frac{c_s r ~\rho }{n}
\label{51}
\end{align}
Substituting \eqref{eq:2bases} in \eqref{52}, we have $\left|\mbf z_{{\bm \alpha}_i}^k \right | \leq \frac{2c_sr \rho }{m}\|\mbf M\|_{\mathcal A,\infty}$. Similarly, we have
\begin{align}
\left |\mathbb E \left[\sum_{i=1}^m |\mbf z^k_{{\bm \alpha}_i}|^2\right] \right | =\frac{m}{n}\sum_{\bm \alpha} |z_{\bm \alpha}^k|^2
\label{53}
\end{align}
Substituting \eqref{52} in \eqref{53}, we have $|\mathbb E [\sum_{i=1}^m |\mbf z^k_{{\bm \alpha}_i}|^2] | = \frac{(2 c_s r \rho)^2}{m n}\|\mbf M\|^2_{\mathcal A,2} $.
We have the necessary terms to bound $\left|\sum_{i=1}^m \mbf z_{{\bm \alpha}_i}^k \right |$ for any $k$, which can bound $\max_{\substack{k}} \left|\sum_{i=1}^m \mbf z_{{\bm \alpha}_i}^k \right |$. 
We apply \cite[Lemma 11]{chen2014robust}, to prove the result.
\end{proof}
\subsection{Upper bound for $\|\mathcal P_{T}^{\perp}(\mbf W)\|$}
\label{ptperpbounddetails}
From the golfing scheme, we have $\|\mathcal P_{T^{\perp}}(\mbf W)\| \leq \sum_{j=1}^{j_0}  \|\mathcal P_{T^{\perp}}\mathcal Q_{\Omega_i} \mathcal P_{T}\bm F_{j-1} \|$. Each of the terms in the right hand side can be bounded as
\begin{eqnarray}\nonumber
\|\mathcal P_{T^{\perp}}\mathcal Q_{\Omega_i} \bm F_{j-1} \| &=&  \|\mathcal P_{T^{\perp}} \bkt{\mathcal Q_{\Omega_i} - \mathcal I}  \mathcal P_{T}\bm F_{j-1}\| \leq   \|\bkt{\mathcal Q_{\Omega_i} - \mathcal I} \bm F_{j-1}\|\\\label{lemma9bound}
 &\leq &   c_2 \brac{\sqrt{\frac{n\log n}{\tilde m}} ~\|\mbf  F_{j-1}\|_{\mathcal A,2} +{\frac{n\log n}{\tilde m}}~ \|\mbf  F_{j-1}\|_{\mathcal A,\infty}}\\\nonumber
 &=&   c_2 \brac{\sqrt{\frac{n\log n}{\tilde m}} ~\|\mathcal{P_T}(\mathcal I-\mathcal Q_{\Omega_i})\mathcal{P}_T(\bm F_{j-2})\|_{\mathcal A,2} + {\frac{n\log n}{\tilde m}}~ \|\mathcal{P_T}(\mathcal I-\mathcal Q_{\Omega_i})\mathcal{P}_T(\bm F_{j-2})\|_{\mathcal A,\infty}}
 \end{eqnarray}
We used Lemma \ref{lemma9} in \eqref{lemma9bound}. Here, $\tilde m = m/j_0 = |\Omega_j|$. Substituting from Lemma \ref{lemma10} and Lemma \ref{lemma11}, we get
\begin{eqnarray}\nonumber
\|\mathcal P_{T^{\perp}}\mathcal Q_{\Omega_i} \bm F_{j-1} \| &\leq&   c_2 \brac{\sqrt{\frac{n\log n}{\tilde m}} ~c_3 \sqrt{\frac{\mu_5\, R\, \log n}{\tilde m}} ~\bkt{\|\mbf F_{j-2}\|_{\mathcal A,2} + \sqrt{\frac{n\, \log n}{\tilde m}}~ \|\mbf F_{j-2}\|_{\mathcal A,\infty}}} +\\\nonumber &&\qquad c_2\brac{{\frac{n\log n}{\tilde m}}~ c_4 \sqrt{\frac{\rho c_s R\log n}{\tilde m}} \sqrt{\frac{\rho c_s R}{n}} ~\|\mbf F_{j-2}\|_{\mathcal A,2} + c_4 ~\frac{\rho c_s R\log n}{\tilde m}~ \|\mbf F_{j-2}\|_{\mathcal A,\infty}}\\\nonumber
&=&  \brac{c_5\bkt{\sqrt{\frac{\mu_5 R\log n}{\tilde m}} + {\frac{\rho c_s R \log n}{\tilde m}}}}^{j_0-1}c_2\brac{ \sqrt{\frac{n\log n}{\tilde m}} ~\|\bm F_{0}\|_{\mathcal A,2} +  {\frac{n\log n}{\tilde m}}~ \|\bm F_{0}\|_{\mathcal A,\infty}}\\
\label{wboundSup}
&\leq&  \bkt{\frac{1}{2}}^{j_0-1}c_2\brac{ \sqrt{\frac{n\log n}{\tilde m}} ~\|\bm F_{0}\|_{\mathcal A,2} + {\frac{n\log n}{\tilde m}}~ \|\bm F_{0}\|_{\mathcal A,\infty}}
 \end{eqnarray}
 The last inequality holds if $\tilde m = m/j_0 \gg \max\bkt{\mu_5,\rho c_s} R \log n $. Substituting for $j_0=3\log_{\frac{1}{\epsilon}}(n)$ assumed in the golfing scheme, we require 
 \begin{equation}
\label{m1condition}
m \gg c_6~\max\bkt{\mu_5,\rho c_s} R \log^2 n
\end{equation}
 to satisfy the above inquality.  We will now focus on bounding $\|\bm F_{0}\|_{\mathcal A,2}$ and $\|\bm F_{0}\|_{\mathcal A,\infty}$, where $\mbf F_0=\mbf U\mbf V^*$.

\subsection{Proof of Lemma \ref{lem20}}
\label{lem20proof}

\begin{proof}
The proof of this theorem is in line with \cite[Lemma 7]{chen2014robust}. The first term is upper bounded by the maximum entry of the matrix (i.e, $\|\mbf U\mbf V^*\|_{\mathcal A,\infty} = \max_{k} |(\mbf U \mbf V^*)_k|$). 
\begin{eqnarray*}
|(\mbf U \mbf V^*)_{\mbf k,l}| &\leq& \abs{\mbf e_{k}^T \mbf \EL(\mbf \EL^*\mbf \EL)^{-\frac{1}{2}}(\mbf \ER^*\mbf \ER)^{-\frac{1}{2}} \mbf \ER \mbf e_l}\\
&\leq&\|\mbf e_{k}^T \mbf \EL\|_{F} \|(\mbf \EL^*\mbf \EL)^{-\frac{1}{2}}\|~\| (\mbf \ER^*\mbf \ER)^{-\frac{1}{2}} \| \|\mbf \ER \mbf e_l\|_F \leq \sqrt{ \frac{\rho ^2}{|\Lambda_1||\Lambda_2|}R} = \frac{\rho \, c_s\, R}n,
\end{eqnarray*}

We now show that the energy of each row of the matrices $\mbf U\mbf V^*$ and $\mathcal P_T\bkt{\sqrt{\omega_{{\bm \alpha}}} ~\bm A_{{\bm \alpha}}}$ are upper bounded; we use this relation to prove \eqref{uvl2} and \eqref{uvl3}. The energy of the rows of $\mbf U\mbf V^*$ are bounded as  
\begin{equation}
\| \mbf e_i^T \mbf U \mbf V^*\|_F^2 = \| \mbf e_i^T \mbf U\|_F^2\leq  \frac{\rho c_s R}{n}\\
\end{equation}

Similarly, by the definition of $\mathcal P_T$, we have
\begin{equation}
\label{uvl3lhs}
\| \mbf e_i^T  \mathcal P_T(\sqrt{\omega_{{\bm \alpha}}}\mbf A_{{\bm \alpha}})\|_F^2 \leq 3\|\mbf e_i^T  \mbf U \mbf U^*(\sqrt{\omega_a}\mbf A_a) \|_F^2 + {3\| \mbf e_i^T (\sqrt{\omega_a}\mbf A_a) \mbf V \mbf V^*\|_F^2} +3\| \mbf e_i^T  \mbf U \mbf U^*(\sqrt{\omega_a}\mbf A_a)\mbf V \mbf V^*\|_F^2
\end{equation}
We now bound each of the terms in the right hand side of the above expression as 
\begin{eqnarray*}
\| \mbf e_i^T \mbf U \mbf U^* (\sqrt{\omega_{\bm \alpha}}\mbf A_{{\bm \alpha}})\|_F^2  &\leq&  \frac{\rho c_s R}{n}\\
\| \mbf e_i^T \mbf U \mbf U^* (\sqrt{\omega_{\bm \alpha}}\mbf A_{{\bm \alpha}})  \mbf V \mbf V^*\|_F^2  &\leq& \| \mbf e_i^T \mbf U\|_F^2 \| \mbf U\|^2 \|(\sqrt{\omega_{\bm \alpha}}\mbf A_{{\bm \alpha}}) \| ^2  \|\mbf V \mbf V^*\|^2 \leq \frac{\rho c_s R}{n}\\
\| \mbf e_i^T  (\sqrt{\omega_{\bm \alpha}}\mbf A_{{\bm \alpha}})\mbf V \mbf V^* \|_F^2  &\leq& \frac{\rho c_s R}{n}
\end{eqnarray*}
Here, we use the property that the operator norm of $\mbf U$ and $\sqrt{\omega_{\bm \alpha}}\mbf A_{{\bm \alpha}}$ are bounded by 1. Substituting in \eqref{uvl3lhs}, we obtain the upper bound for the energy of the rows as 
\begin{align*}
\| \mbf e_i^T  \mathcal P_T(\sqrt{\omega_a}\mbf A_a)\|_F^2 \leq \frac{\mu_3 c_sR}{n}. 
\end{align*}

Now, applying \cite[Lemma 12]{chen2014robust} that relates the upper bound of $\|\bm M\|_{\mathcal A,2}^2$ to the upper bound of the energy of the rows $\max_i\|\bm e_i \bm M\|^2$, we obtain the results \eqref{uvl2} and \eqref{uvl3}.

\end{proof}

Substituting \eqref{uvinfty} and \eqref{uvl2} in \eqref{wboundSup}, we get
\begin{eqnarray}\nonumber
\|\mathcal P_{T^{\perp}}\mathcal Q_{\Omega_i} \bm F_{j-1} \| 
&\leq&  \bkt{\frac{1}{2}}^{j_0-1}c_2\brac{ \sqrt{\frac{n\log n}{\tilde m}} ~ \sqrt{\frac{c_7 \mu_3 c_s R \log^2 (n)}{n}} + {\frac{n\log n}{\tilde m}}~ \frac{\rho c_s R}{n}}\\
&=& \bkt{\frac{1}{2}}^{j_0-1}c_2\brac{ \sqrt{\frac{c_7 \mu_3 c_s R \log^3 (n)}{\tilde m}} +  {\frac{\rho c_s R\log n}{\tilde m}}}
 \end{eqnarray}
Similar to the argument before, if $\tilde m = m/j_0 \gg c_7 \mu_3 c_s R \log^3 n$ and $\tilde m = m/j_0 \gg \rho c_s R \log n$, or equivalently $m \gg {\rho c_s} r \log^2 n$, we have $\|\mathcal P_{T^{\perp}}\mathcal Q_{\Omega_i} \bm F_{j-1} \| \leq \left(\frac{1}{2}\right)^{j_0}$. Combining these conditions with \eqref{m1condition}, we need 
\begin{equation}
m \gg ~\max\bkt{c_6\mu_5\,R \log^2 n, c_6 \rho \,c_s\,R \log^2 n,c_7 \mu_3 c_s R \log^3 n,\ {\rho c_s} r \log^2 n}.
\end{equation}
 From \eqref{uvl3}, we see that the constant $\mu_5 $ in Lemma \ref{lemma11} can be chosen as 
\begin{equation}
\label{mu5}
\mu_5 = c_7~ \mu_3~ c_s ~\log^2(n)
\end{equation}
such that  $\omega_{\mbf k} \|\mathcal P_T(\mbf A_{\mbf k})\|_{\mathcal A,2} \leq \frac{\mu_5 R}{n}$. Substituting for $\mu_5$ from \eqref{mu5}, we observe that the dominant term is the first one due to its dependence on $\log^4(n)$. Thus, $\|\mathcal P_{T^{\perp}}\mathcal Q_{\Omega_i} \bm F_{j-1} \|  < 1/2$ if 
\begin{equation}
m> c_6  c_7\, c_s\, \bkt{3\rho} \,R\,\log^4(n).
\end{equation}
\section{Proof of Theorem 5}
\label{noisythmproof}
The proof is similar to that in Appendix L of \cite{chen2014robust}.
The optimization problem \eqref{eq:noisystlr} can be reformulated in the lifted matrix domain as:
\begin{equation}
\label{noisystlr}
\mbox{minimize}_{\bm X}~ \left\|\bm X\right\|_* ~\mbox{subject to}~
\left\| \mathcal A_{\Theta}(\bm X) - \mathcal  A_{\Theta}(\Tf)\right \|_F \leq \delta; ~{\mathcal A}^{\perp}(\mbf X)=\mbf 0.
\end{equation}
Let $\widehat{\mbf X} = \Tf + \mbf H$ be the solution to \eqref{noisystlr}, where $\mbf H = \mathcal{A}_\Omega (\mbf H) + \mathcal{A}_{\Omega^\perp} (\mbf H)$. Using \eqref{noisystlr}, 
\begin{align}
\| \mbf A_\Omega \mbf H\|_F \leq \| \mbf A_\Omega(\widehat{\mbf X} -\Tfn ) \|_F + \| \mbf A_\Omega(\Tf -\Tfn ) \|_F \leq 2 \sqrt{n} \delta  \nonumber
\end{align}
From the argument in \eqref{infeasibility}, we only consider the case,
\begin{align}
\| \mathcal{P}_T \mathcal{A}_{\Omega^\perp} (\mbf H)\|_F \leq 2n \| \mathcal{P}_{T{^\perp}} \mathcal{A}_{\Omega^{\perp}} (\mbf H)\|_F 
\label{case2eqn}
\end{align}  
defined in \ref{case2} Case 2. 
Now, $\|\Tf\|_*  \geq \|\widehat{\mbf X}\|_* \geq \|\Tf +\mathcal{A}_{\Omega^\perp} (\mbf H)\|_* - \|\mathcal{A}_\Omega (\mbf H) \|_* $. From \eqref{one6} and the treatement in Appendix L, \cite{chen2014robust}, we have 
$\|\Tf\|_*  \geq \|\Tf\|_* + \frac{1}{6} \|\mathcal{P}_T \mathcal{A}_{\Omega^\perp} (\mbf H) \|_F - \|\mathcal{A}_{\Omega} (\mbf H) \|_*$, which gives $\|\mathcal{P}_T \mathcal{A}_{\Omega^\perp} (\mbf H) \|_F \leq 6\|\mathcal{A}_{\Omega} (\mbf H) \|_* \leq 6 \sqrt{n} \|\mathcal{A}_{\Omega} (\mbf H) \|_F \leq 12n \delta $. Finally using \eqref{case2eqn}, 
\begin{align}
\| \mbf H \|_F \leq \|\mathcal{A}_{\Omega} (\mbf H) \|_F + \| \mathcal{P}_T \mathcal{A}_{\Omega^\perp} (\mbf H)\|_F + \| \mathcal{P}_{T{^\perp}} \mathcal{A}_{\Omega^{\perp}} (\mbf H)\|_F \leq 2 \sqrt{n} \delta + 12n \delta +24 n^2 \delta \leq c n^2 \delta \nonumber
\end{align}
which concludes the proof.

\end{document}